\documentclass[11pt]{amsproc}

\usepackage{array,amsmath,amssymb,amsthm,color,tikz,vmargin,overpic,amstext}
\usepackage{amsmath,amsthm,amsfonts,amscd,amssymb,amstext}
\usepackage{mathtools}

\usepackage{booktabs}
\usepackage{esint}
\usepackage{relsize}
\usepackage{graphicx,color}
\usepackage{tabularx}
\usepackage{titletoc}
\usepackage{microtype}
\usepackage{caption}
\usepackage{subcaption}
\usepackage{listings}
\usepackage{cite}
\usepackage{url}
\usepackage{alltt}
\usepackage{bbm}

\newtheorem{theorem}{Theorem}

\newtheorem{remark}[theorem]{Remark}

\newtheorem{lemma}[theorem]{Lemma}
\newtheorem{proposition}[theorem]{Proposition}
\newtheorem{corollary}{Corollary}
\newtheorem{conjecture}{Conjecture}
\numberwithin{equation}{section}
\numberwithin{theorem}{section}

\title{Characteristic polynomials of complex random matrices and Painlev\'e transcendents}

\title[Complex random matrices and Painlev\'e transcendents]{Characteristic polynomials of complex random matrices and Painlev\'e transcendents}

\author{Alfredo Dea\~{n}o}
\address[A. Dea\~{n}o]{\newline School of Mathematics, Statistics and Actuarial Science. University of Kent, Canterbury CT2 7NF, United Kingdom. \textit{A.Deano-Cabrera@kent.ac.uk}}

\author{Nick Simm}
\address[N. Simm]{\newline Mathematics Department. University of Sussex, Falmer Campus, Brighton, BN1 9RH, United Kingdom. \textit{n.j.simm@sussex.ac.uk}}
\date{\today}

\begin{document}

\begin{abstract}
We study expectations of powers and correlation functions for characteristic polynomials of $N \times N$ non-Hermitian random matrices. For the $1$-point and $2$-point correlation function, we obtain several characterizations in terms of Painlev\'e transcendents, both at finite-$N$ and asymptotically as $N \to \infty$. In the asymptotic analysis, two regimes of interest are distinguished: boundary asymptotics where parameters of the correlation function can touch the boundary of the limiting eigenvalue support and bulk asymptotics where they are strictly inside the support. For the complex Ginibre ensemble this involves Painlev\'e IV at the boundary as $N \to \infty$. Our approach, together with the results in \cite{HW17} suggests that this should arise in a much broader class of planar models. For the bulk asymptotics, one of our results can be interpreted as the merging of two `planar Fisher-Hartwig singularities' where Painlev\'e V arises in the asymptotics. We also discuss the correspondence of our results with a normal matrix model with $d$-fold rotational symmetries known as the \textit{lemniscate ensemble}, recently studied in \cite{BGM17, BGG18}. Our approach is flexible enough to apply to non-Gaussian models such as the truncated unitary ensemble or induced Ginibre ensemble; we show that in the former case Painlev\'e VI arises at finite-$N$. Scaling near the boundary leads to Painlev\'e V, in contrast to the Ginibre ensemble.
\end{abstract} 

\maketitle

\section{Introduction and results}
The purpose of this paper is to study correlation functions of characteristic polynomials, of the form
\begin{equation}
R_{\vec{\gamma}}(\vec{z}) := \mathbb{E}\left(\prod_{i=1}^{m}|\det(G_{N}-z_{i})|^{\gamma_{i}}\right) \label{correlationfn}
\end{equation}
where $\vec{\gamma} = (\gamma_1,\ldots,\gamma_{m})$ are parameters\footnote{We mainly consider the case that $\vec{\gamma} = 2\vec{k}$ where $\vec{k} \in \mathbb{N}^{m}$, though not exclusively.} and $\vec{z} \in \mathbb{C}^{m}$. In the first instance we take $G_{N}$ to be a standard complex Ginibre random matrix, that is $G_{N} = \frac{1}{\sqrt{N}}\,G$ where $G$ is an $N \times N$ matrix of \textit{i.i.d.} standard complex normal random variables, though we shall also discuss other examples. 

Particular attention will be paid to the cases $m=1$ and $m=2$, where the average \eqref{correlationfn} turns out to be intimately related to Painlev\'e transcendents. These are solutions of a distinguished class of non-linear second order differential equations, labelled Painlev\'e I - Painlev\'e VI, see for instance \cite[Chapter 32]{NIST:DLMF}. The relevance of Painlev\'e transcendents in random matrix theory arose most famously in the work of Tracy and Widom in connection with the largest eigenvalue of a Hermitian random matrix \cite{TW94}, building on the earlier developments of Jimbo, Miwa, M\^{o}ri and Sato \cite{JMMS80}. They are now among the most important special functions of mathematical physics, and they appear in connection with reduction of integrable PDEs, models in statistical mechanics, combinatorics and orthogonal polynomials, to name only a few applications, see \cite{WVA19} for a recent survey, and the text \cite{FIKN}. 

The characterization of averages involving Hermitian (or unitary) random matrices in terms of Painlev\'e transcendents has attracted considerable activity in the last decades. However, the fact that Painlev\'e transcendents also play an important role in the context of non-Hermitian random matrices appears to be somewhat less appreciated and will be the main focus of this work. For the complex Ginibre ensemble, this has been discussed to some extent in the physics literature: Nishigaki and Kamenev \cite{NK02} applied methods from supersymmetry to compute \eqref{correlationfn} when $m=1$ and a special case of the $m=2$ correlation function. In subsequent work of Kanzieper \cite{Kan05}, in the context of replica field theories, it was pointed out that the results in \cite{NK02} imply a relation to Painlev\'e IV and Painlev\'e V transcendents. Here, in addition to studying the general correlations \eqref{correlationfn}, we will arrive at these results via what we believe to be a mathematically more transparent approach, circumventing the need for supersymmetry techniques. Another advantage of our approach is its applicability to other classes of planar models, beyond the complex Ginibre ensemble.

In the case of Hermitian (or unitary) random matrices, correlation functions of characteristic polynomials of type \eqref{correlationfn} were studied quite intensively over the last two decades. The pioneering work of Keating and Snaith \cite{KS00} demonstrated their importance in modelling statistical properties of the Riemann zeta function high up on the critical line. For Hermitian random matrices, higher order correlations were studied by Br\'ezin and Hikami \cite{BH00,BH00edge} in various asymptotic regimes and the importance of determinantal structures was emphasized. In the case of non-integer exponents and GUE matrices, Krasovsky \cite{Kra07} applied Riemann-Hilbert techniques to calculate the corresponding asymptotics. More recently there has been a resurgence of interest in characteristic polynomials of random matrices due to their association with log-correlated Gaussian fields \cite{HKO01,FK14,W15,FKS16}. Along these lines, quantity \eqref{correlationfn} can be interpteted as the multi-point Laplace transform of the \textit{random process} $\phi_{N}(z) = \log|\det(G_{N}-z)|$ which is asymptotically Gaussian and logarithmically correlated, see \cite{RV07,AHM11,BWW18,WW18}. This is also closely related to the recently studied question of `moments of moments' \cite{BK19,AK19,Fahs19} (where Painlev\'e V also arises), albeit in the Hermitian context. In the present non-Hermitian setting the understanding of \eqref{correlationfn} remains at an early stage, especially at the level of asymptotic analysis. 

Let us now be more precise about the asymptotic regimes under consideration. Under the chosen normalization of the matrices $G_{N}$, we have the famous \textit{circular law}: in the limit $N \to \infty$ the eigenvalues of $G_{N}$ are uniformly distributed on the unit disc in the complex plane. Our focus will be on large $N$ asymptotics of \eqref{correlationfn} in the microscopic regime, that is where the points $\{z_{i}\}_{i=1}^{m}$ have a separation on the order of the mean eigenvalue spacing of $G_{N}$ of size $1/\sqrt{N}$. This is to be distinguished from the macroscopic or mesoscopic regimes which involve separations at larger scales. The possible microscopic scalings of \eqref{correlationfn} are then further divided into \textit{boundary asymptotics} where the $z_{i}$'s are close to the boundary of the eigenvalue support (close to the unit circle for the matrices $G_{N}$) and \textit{bulk asymptotics} where the $z_{i}$'s are strictly inside the support. The boundary asymptotics, although interesting in their own right, have an application to the normal matrix model which we now discuss.

Recall that the normal matrix model is the probability measure on $N$ points $\lambda_{1},\ldots,\lambda_{N}$ in the complex plane defined by
\begin{equation}
d\mathbb{P}(\lambda_1,\lambda_2,\ldots,\lambda_N) = \frac{1}{\mathcal{Z}_{N}} \prod_{j=1}^{N}e^{-NV(\lambda_j)}\prod_{1 \leq i < j \leq N}|\lambda_{j}-\lambda_{i}|^{2}\,d^{2}\lambda_1\ldots d^{2}\lambda_N, \label{law}
\end{equation}
where $\mathcal{Z}_{N}$ is a normalization constant, known as the \textit{partition function}, defined by
\begin{equation}
\mathcal{Z}_{N} = \int_{\mathbb{C}^{N}}\prod_{j=1}^{N}e^{-NV(\lambda_j)}\prod_{1 \leq i < j \leq N}|\lambda_{j}-\lambda_{i}|^{2}\,d^{2}\lambda_1\ldots d^{2}\lambda_N. \label{generalpf}
\end{equation}
Although the matrix $G_{N}$ is almost surely not normal, its eigenvalue distribution is known to take the form \eqref{law} with $V(\lambda) = |\lambda|^{2}$. Therefore, the eigenvalues of the complex Ginibre ensemble can be viewed as the most basic yet non-trivial example of the model \eqref{law}. In general, the function $V(\lambda)$ is called the external potential and subject to suitable growth and regularity conditions, in the limit $N \to \infty$ the random points $\{\lambda_{j}\}_{j=1}^{N}$ are supported on a compact subset of the plane known as the \textit{droplet}, denoted $\mathcal{S}$. The density of eigenvalues in $\mathcal{S}$ is described by an equilibrium measure coming from weighted potential theory, see \cite{AKM19} for further background. In the electrostatic interpretation of the model \eqref{law}, quantity \eqref{correlationfn} can be viewed as a partition function of the form \eqref{generalpf} subjected to the insertion of $m$ \textit{point charges} at locations $\vec{z}$ and strengths $\vec{\gamma}$, see \cite{SYM17,AKS18,LY19}.

Recently, there has been interest in studying the model \eqref{law} corresponding to the class of potentials
\begin{equation}
V^{(d)}(\lambda,t) = |\lambda|^{2d}-t(\lambda^{d}+\overline{\lambda}^{d}), \qquad d \in \mathbb{N}, \quad t \in \mathbb{R}, \label{lem}
\end{equation}
\begin{figure}
\begin{center}
\includegraphics[scale=0.8]{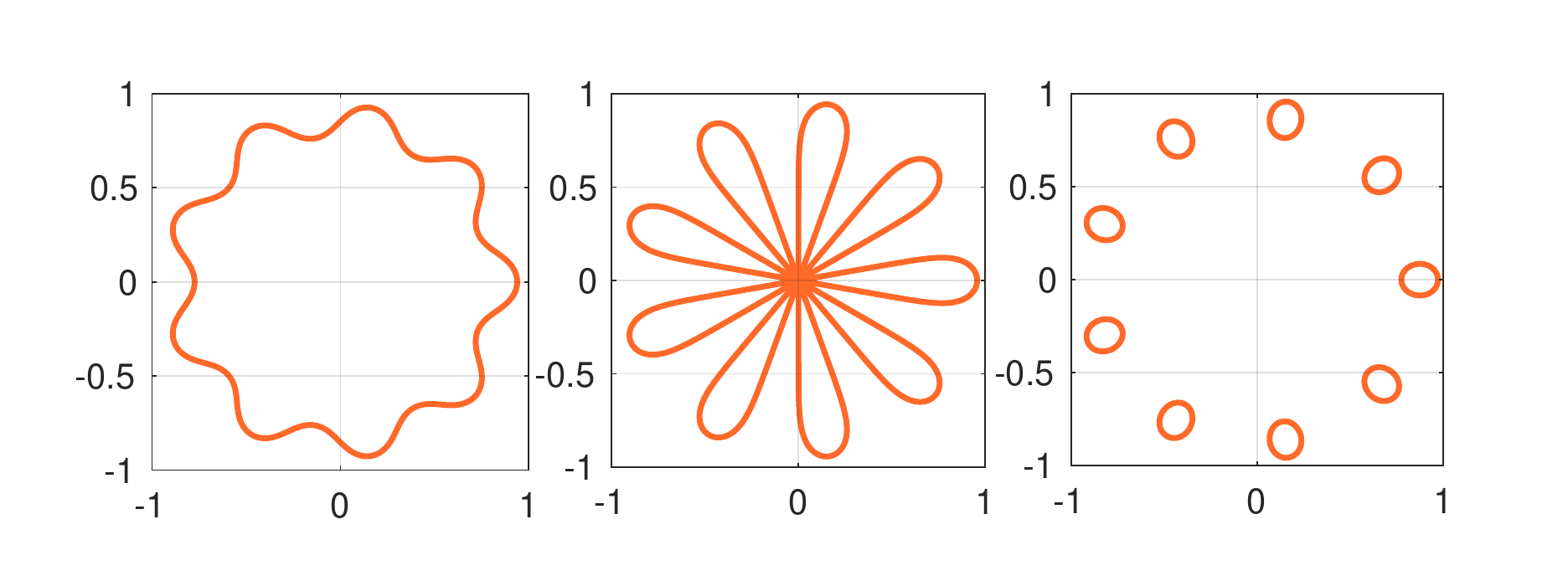}
\caption{The boundary of the droplet for the lemniscate potential \eqref{lem} for $d=9$ with critical value $t_{\mathrm{c}} = 1/\sqrt{d}$. From left to right the plots show $t < t_{\mathrm{c}}$, $t=t_{\mathrm{c}}$ and $t>t_{\mathrm{c}}$.}
\label{fig:lemfig}
\end{center}
\end{figure}
which are invariant under the $d$-fold symmetry $\lambda \to \lambda e^{2\pi i /d}$. This model was introduced by Balogh and Merzi \cite{BM15} who showed that the equilibrium measure is supported on the interior of the domain, $\mathcal{S} = \{\lambda : |\lambda^{d}-t| \leq t_{\mathrm{c}}\}$ where $t_{\mathrm{c}} = d^{-1/2}$. Due to the lemniscate type shape of $\mathcal{S}$, the random point configuration corresponding to \eqref{lem} has been referred to as the \textit{lemniscate ensemble} \cite{AKMW15}, see Figure \ref{fig:lemfig}. When $t$ passes through $t_{\mathrm{c}}$, the topology of the droplet undergoes a transition from being simply connected to consisting of $d$ connected components. The case $t=t_{\mathrm{c}}$ has been singled out as a model with a non-trivial boundary singularity (see \cite{AKMW15}) and the effect of this singularity on the partition function asymptotics is of interest (see also Remark \ref{rem:serfaty}). 

In the recent work \cite{BGG18} the orthogonal polynomials for the planar weight $e^{-NV^{(d)}(\lambda,t)}$ have been analysed. Close to the transition $t=t_{\mathrm{c}}$ the asymptotics were expressed in terms of the Hamiltonian of the Painlev\'e IV system. This followed earlier works \cite{SYM17, BGM17} on the off-critical cases $t \neq t_{\mathrm{c}}$ which were expressed in terms of classical special functions, see also \cite{BBLM15}. A natural question is to look at the partition function associated with \eqref{lem}. The following lemma shows its relation to characteristic polynomials of the Ginibre ensemble.
\begin{lemma}
\label{lem:introredgin}
Let $\mathcal{Z}^{(\mathrm{Lem}_{d})}_{N}(t)$ denote the partition function corresponding to \eqref{lem}. Then we have the identity
\begin{equation}
\mathcal{Z}^{(\mathrm{Lem}_{d})}_{Nd}(t) = e^{(Ntd)^{2}}\tilde{c}_{N,d}\prod_{\ell=0}^{d-1}\mathbb{E}\left(|\det(G_{N}-t\sqrt{d})|^{\gamma_\ell}\right) \label{ginred}
\end{equation}
where
\begin{equation}
\gamma_{l} = -2\left(1-\frac{\ell+1}{d}\right), \qquad \ell=0,\ldots,d-1. \label{gamell}
\end{equation}
Here $\tilde{c}_{N,d}$ is an explicit constant that we give later in Section \ref{se:lem}.
\end{lemma}
Notice that double scaling the parameter $t$ close to $t_{\mathrm{c}} := \frac{1}{\sqrt{d}}$ in \eqref{ginred} corresponds precisely to boundary asymptotics of \eqref{correlationfn} with $m=1$. In what follows we will present the boundary asymptotics in the Ginibre case; consequences for \eqref{ginred} are given in Section \ref{se:lem}. The best result we are aware of for the average in \eqref{ginred} holds strictly away from the boundary and the authors do not comment on the connection with the lemniscate ensemble.
\begin{theorem}[Webb and Wong \cite{WW18}]
\label{th:ww}
Fix $\delta>0$ and consider the closed disc $D_{1-\delta}$ of radius $1-\delta$. Then provided $\mathrm{Re}(\gamma)>-2$, the following asymptotic formula holds uniformly for $z \in D_{1-\delta}$
\begin{equation}
\mathbb{E}\left(|\det(G_{N}-z)|^{\gamma}\right) = e^{\frac{N\gamma}{2}\,(|z|^{2}-1)+\frac{\gamma^{2}}{8}\,\log(N)}\,\frac{(2\pi)^{\frac{\gamma}{4}}}{G(1+\frac{\gamma}{2})}(1+o(1)), \qquad N \to \infty, \label{webbwongform}
\end{equation}
where $G(z)$ is the Barnes G-function. The expansion \eqref{webbwongform} holds uniformly in compact subsets of $\mathrm{Re}(\gamma)>-2$.
\end{theorem}
We now state our main results.

\subsection{Complex Ginibre ensemble: Boundary asymptotics}
Our first result extends the asymptotic expansion \eqref{webbwongform} uniformly to the boundary $|z|=1$ and leads to a new behaviour in the constant term. In order to state our result, recall that a \textit{GUE (Gaussian Unitary Ensemble) random matrix} is a Hermitian random matrix $H$ whose density is proportional to $e^{-\mathrm{Tr}(H^{2})/2}$ with respect to the Lebesgue measure on its functionally independent entries. We denote by $\lambda^{(\mathrm{GUE}_{k})}_{\mathrm{max}}$ the largest eigenvalue of the matrix $H$, having size $k \times k$.

\begin{theorem}
\label{th:onecharge}
Let $z$ belong to the closed disc $D_{R}$ of radius $R = 1+L/\sqrt{N}$ for some positive constant $L>0$. Then for any $k \in \mathbb{N}$, the following asymptotic expansion holds uniformly on $D_{R}$,
\begin{equation}
\begin{split}
\mathbb{E}\left(|\det(G_{N}-z)|^{2k}\right) = e^{Nk(|z|^{2}-1)+\frac{k^{2}}{2}\, \log(N)}\,&\frac{(2\pi)^{\frac{k}{2}}}{G(1+k)}\,F_{k}(\sqrt{N}(1-|z|^{2}))\\
\times\left(1+o(1)\right), \qquad N \to \infty, \label{ginexpan}
\end{split}
\end{equation}
where $F_{k}(x)$ is the distribution function of the largest eigenvalue of a $k \times k$ GUE random matrix,
\begin{equation}
F_{k}(x) = \mathbb{P}(\lambda^{(\mathrm{GUE}_{k})}_{\mathrm{max}}<x). \label{boundaryerfc}
\end{equation}
\end{theorem}
\begin{proof}
See Section \ref{se:dualitylue}.
\end{proof}
\begin{remark}
A result of Tracy and Widom characterizes the above GUE probability in terms of Painlev\'e transcendents. In particular in \cite{TW94} we find
\begin{equation}
F_{k}(x) = \mathrm{exp}\left(-\int_{x}^{\infty}\sigma^{(\mathrm{IV})}_{k}(t)\,dt\right), \label{Fk}
\end{equation}
where $\sigma^{(\mathrm{IV})}_{k}(t)$ satisfies the Jimbo-Miwa-Okamoto $\sigma$-form of the Painlev\'e IV equation
\begin{equation}
(\sigma'')^{2}+4(\sigma')^{2}(\sigma'+k)-(t\sigma'-\sigma)^{2} = 0 \label{p4}
\end{equation}
subject to the boundary condition $\sigma(t) = -kt-k^{2}/t + o(1/t)$ as $t \to -\infty$ (see \cite[Equation 4.18]{FWPII}). 
\end{remark}
We expect that the form of expansion \eqref{ginexpan} also holds more generally when $k = \frac{\gamma}{2}$ and $\gamma$ satisfies the hypothesis of Theorem \ref{th:ww}, noting that the above Painlev\'e characterization does not require integrality of $k$. In the later sections, we will provide some heuristic computations lending support to such an extrapolation to non-integer $k$. This is based on the fact that for finite $N$, the left-hand side of \eqref{ginexpan} can be expressed in terms of a Painlev\'e V transcendent for general $k = \frac{\gamma}{2}$, see Section \ref{se:redgroup}. The Painlev\'e IV on the right-hand side of \eqref{ginexpan} then arises from a rescaling (confluent limit) of this equation.
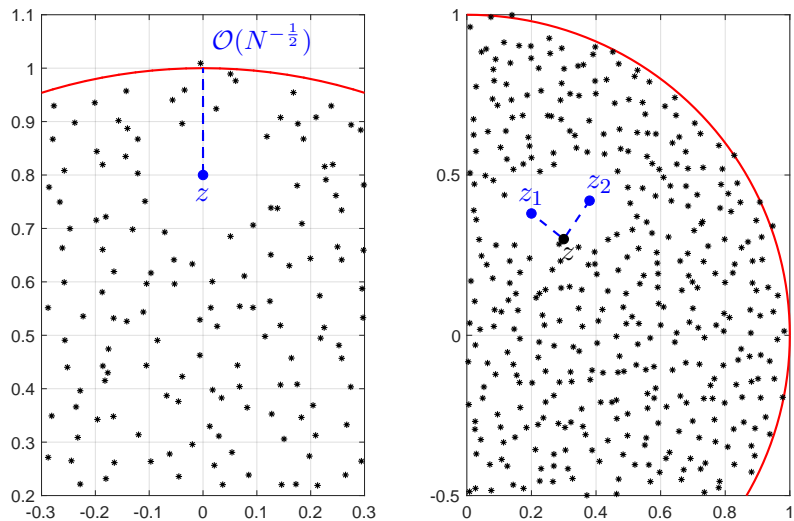
\begin{figure}
\begin{center}
\begin{overpic}[scale=.4]{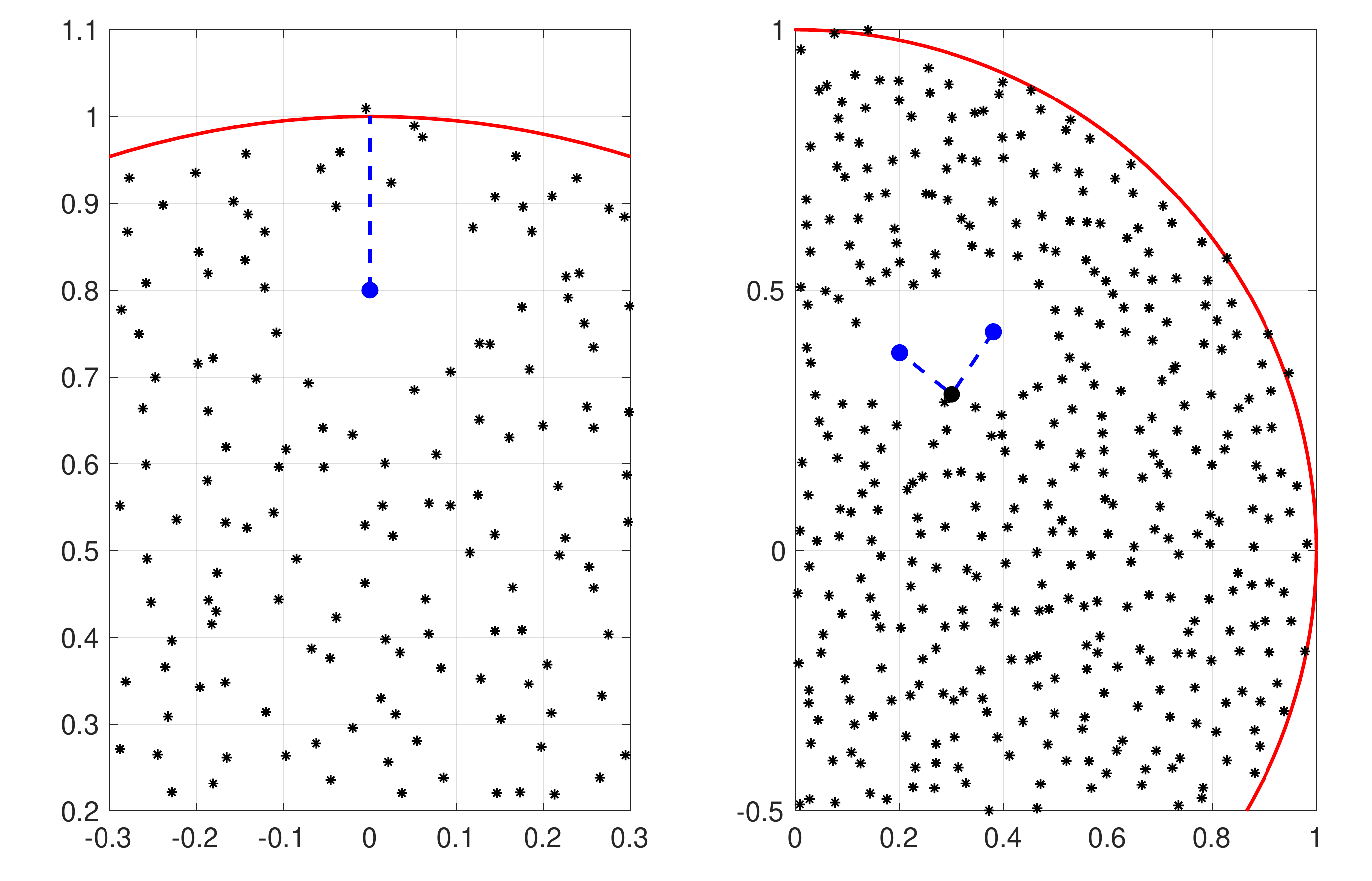}
\put(28,58){\color{blue}\small{$\mathcal{O}(N^{-\frac{1}{2}})$}}
\put(26,40){\color{blue}{$z$}}
\put(64,40){\color{blue}{$z_1$}}
\put(72,41.5){\color{blue}{$z_2$}}
\put(69,33){\color{black}{$z$}}
\end{overpic}
\caption{Schematic insertion of a charge near the boundary of distance on the order of $N^{-1/2}$ (left), and of two charges $z_1$ and $z_2$ of distance on the order of $N^{-1/2}$ from a fixed point $z$ in the bulk of the complex Ginibre ensemble (right). In both cases $N=1000$, and the unit circle is depicted in red.}
\label{fig_Ginibre1}
\end{center}
\end{figure}
Notice that when $|z|<1$ is fixed inside the unit disc and $N$ is large, we can replace the GUE probability in \eqref{boundaryerfc} with $1$ and we recover precisely \eqref{webbwongform} with $\gamma=2k$. On the other hand, if we are close to the boundary with $|z| = 1-u/\sqrt{N}$ with $u$ fixed, the GUE probability is not negligible and converges to $F_{k}(2u)$. When $|z|>1+\delta$ with $\delta>0$ fixed, a different type of expansion holds, see Appendix \ref{app:gff} for details. 

As we already mentioned, Painlev\'e IV type objects also appeared in the recent work \cite{BGG18}, which (taking into account relation \eqref{ginred}) is likely to be related with our result here, see Section \ref{se:lem}. Analogous results on the line are known: the work \cite{BCI16} considers a jump singularity colliding with the spectral edge of GUE random matrices, which they describe in terms of Painlev\'e II.

\subsection{Complex Ginibre ensemble: Bulk collision of two singularities}
\label{se:ginibre}
Consider the particular case $m=2$ of \eqref{correlationfn} and set
\begin{equation}
z_{1} = z + \frac{u_1}{\sqrt{N}}, \qquad z_{2} = z + \frac{u_{2}}{\sqrt{N}}, \label{z1z2scaling}
\end{equation}
where $z$ is a point in the bulk, \textit{i.e.} $|z| < 1-\delta$ for some fixed $\delta>0$. That is, we consider the case of two singularities merging (or colliding) inside the bulk of the spectrum, see Figure \ref{fig_Ginibre1}. The corresponding asymptotics of \eqref{correlationfn} turns out to be related to the LUE (Laguerre Unitary Ensemble) defined as follows: Let $G_{k,p}$ denote a rectangular $k \times p$ matrix of \textit{i.i.d.} independent standard complex normal random variables with $p\geq k$. Then the matrix $W = G_{k,p}G^{\dagger}_{k,p}$ is said to be distributed according to the LUE with $p$ degrees of freedom, denoted $\mathrm{LUE}_{k,p}$, though sometimes it is convenient to refer to the parameter $\alpha = p-k$. The names \textit{Wishart matrix} or \textit{sample covariance matrix} are also in common use, see the text \cite{Forrester_loggas} for further background. The largest eigenvalue of $W$ will be denoted $\lambda^{(\mathrm{LUE}_{k,p})}_{\mathrm{max}}$.
\begin{theorem}
\label{th:twocharge}
Let $k_{1},k_{2} \in \mathbb{N}$, where we assume without loss of generality that $k_1 \geq k_2$ and take $z_1,z_2$ scaled according to \eqref{z1z2scaling}. Then the following holds uniformly for $u_1$ and $u_2$ varying in compact subsets of $\mathbb{C}$,
\begin{equation}
\begin{split}
\mathbb{E}&\left(|\det(G_{N}-z_1)|^{2k_1}|\det(G_{N}-z_2)|^{2k_2}\right) = e^{k_{1}N(|z_1|^{2}-1)+k_{2}N(|z_2|^{2}-1)+\frac{k_1^{2}}{2}\log(N)+\frac{k_{2}^{2}}{2}\log(N)}\\
&\times e^{-2k_1k_2\log|z_2-z_1|}\,\frac{(2\pi)^{(k_1+k_2)/2}}{G(1+k_1)G(1+k_2)}F_{k_1,k_2}(N|z_2-z_1|^{2})(1+o(1)), \qquad N \to \infty, \label{corr2}
\end{split}
\end{equation}
where $F_{k_1,k_2}(x)$ is the distribution function of the largest eigenvalue of a $k_2 \times k_2$ LUE matrix with $k_1$ degrees of freedom:
\begin{equation}
F_{k_1,k_2}(x) = \mathbb{P}\left(\lambda_{\mathrm{max}}^{(\mathrm{LUE}_{k_2,k_1})} < x\right). \label{LUEprob}
\end{equation}
\end{theorem}
\begin{proof}
See Section \ref{se:bulk}.
\end{proof}
\begin{remark}
As with Theorem \ref{th:onecharge}, the probability $F_{k_1,k_2}(x)$ has a characterization in terms of a non-linear second order differential equation. In the same work of Tracy and Widom \cite{TW94} one finds the result
\begin{equation}
F_{k_1,k_2}(x) = \mathrm{exp}\left(-\int_{x}^{\infty}\frac{\sigma^{(\mathrm{V})}_{k_2,k_1-k_2}(t)}{t}\,dt\right), \label{LUEfk1k2}
\end{equation}
where $\sigma^{(\mathrm{V})}_{k,\alpha}(t)$ satisfies the Jimbo-Miwa-Okamoto $\sigma$-form of the Painlev\'e V equation,
\begin{equation}
(t\sigma'')^{2}-[\sigma-t\sigma'+2(\sigma')^{2}+(2k+\alpha)\sigma']^{2}+4(\sigma')^{2}(k+\alpha+\sigma')(k+\sigma')=0. \label{pvintro}
\end{equation}
\end{remark}
If we consider a limiting case where the separation $|u_{2}-u_{1}| \to \infty$ we can replace the LUE probability with $1$. The remaining terms in the asymptotic expansion \eqref{corr2} are consistent with a conjectured formula of Webb and Wong \cite{WW18} for the \textit{global asymptotics} of \eqref{correlationfn} (defined by a separation between $z_{2}$ and $z_{1}$ of order $O(1)$). Indeed, the appearance of the term $-2k_{1}k_{2}\log|z_{2}-z_{1}|$ should be interpreted as the covariance function of the random process $\phi_{N}(z) := \log|\det(G_{N}-z)|$. Such \textit{logarithmic correlations} are widely anticipated to appear at global or mesoscopic regimes because of the association of $\phi_{N}(z)$ with the Gaussian free field \cite{RV07, AHM11}. 

We remark that the particular case of \eqref{corr2} with $k_1=k_2$ and $z=0$ appeared in the mentioned work of Nishigaki and Kamenev \cite{NK02} using a technique based on Grassmann integration, but the terms of order $\log(N)$ appearing in \eqref{corr2} appear to be absent in \cite{NK02}. This could be due to their intended application to replica limits which could render such terms negligible in their context. Finally we remark that a one-dimensional version of this merging of singularities has been investigated in the context of Toeplitz and Hankel determinants, with non-integer exponents, and also involves Painlev\'e V \cite{CK15,CF16}.

As with Theorem \ref{th:onecharge}, we believe that Theorem \ref{th:twocharge} continues to hold when $k_1$ and $k_2$ are general real numbers such that the left-hand side of \eqref{corr2} is finite, with equation \eqref{pvintro} defining the extrapolation to such non-integral exponents. To give support to this we give an example involving (partially) non-integer exponents, as follows.
\begin{theorem}
\label{th:nonintegerbulk}
Let $k_{2} \in \mathbb{N}$ and $\gamma$ be a real number such that $\gamma \geq k_{2}$. Consider the left-hand side of \eqref{corr2} scaled according to \eqref{z1z2scaling} at the center of the eigenvalue support, $z=0$, with $u_{1}=0$ and set $k_{1}=\gamma$. Then the expansion on the right-hand side of \eqref{corr2} remains true with $k_{1} = \gamma$.
\end{theorem}
\begin{proof}
See Section \ref{se:bulk}.
\end{proof}
As we explain in Section \ref{se:bulk}, Theorem \ref{th:nonintegerbulk} can be viewed as an asymptotic expansion for the $m=1$ case of \eqref{correlationfn}, with $G_{N}$ replaced by a certain non-Gaussian matrix: the \textit{induced Ginibre ensemble} introduced and studied in \cite{FBKSZ}, see equation \eqref{corr2proved}. 
\subsection{Truncated unitary ensemble and Painlev\'e transcendents}
Let $M$ be a positive integer and consider $U(M)$, the group of $M \times M$ unitary matrices. The CUE (Circular Unitary Ensemble) is $U(M)$ equipped with the normalized Haar measure $d\mu(U)$, where $U \in U(M)$. A \textit{truncation} of $U$ is simply a sub-block, which by invariance of the Haar measure we can take to be the upper-left $N \times N$ sub-block of $U$, that we denote here by $T$ and assume $N < M$. The eigenvalues $\lambda_1,\ldots,\lambda_N$ of the sub-unitary matrix $T$ lie strictly inside the unit disc $|\lambda_{j}| \leq 1$ for all $j=1,\ldots,N$. Their joint probability density function has been computed by \.{Z}yczkowski and Sommers \cite{ZS00} as 
\begin{equation}
P(\lambda_1,\ldots,\lambda_N) \propto \prod_{j=1}^{N}(1-|\lambda_{j}|^{2})^{M-N-1}\prod_{1 \leq i < j \leq N}|\lambda_{j}-\lambda_{i}|^{2}. \label{zspdf}
\end{equation}
Using \eqref{zspdf} the authors showed that statistics of the eigenvalues depend sensitively on the relative scaling of $N$ and $M$ as $M \to \infty$. Roughly speaking, unless $N$ is very close to $M$, the matrix $T$ is expected to have similar spectral characteristics as the Ginibre matrix $G_{N}$ from the previous sections. However, when $M-N = \kappa = O(1)$, so that only a constant order of rows and columns are truncated, one expects new behaviour. This is known as the \textit{weak non-unitarity limit}, see \cite{FS03,KS11} for further background and applications.

\begin{theorem}[Boundary asymptotics]
\label{th:tcueonecharge}
Let $M-N = \kappa$ be a fixed positive integer and fix $k \in \mathbb{N}$. Consider the boundary scaling
\begin{equation}
|z| = 1-\frac{u}{N}, \qquad u>0.
\end{equation}
Then the following asymptotic formula holds
\begin{equation}
\mathbb{E}\left(|\det(T-z)|^{2k}\right) = N^{k^{2}}\,\frac{G(k+\kappa+1)}{G(k+1)G(\kappa+1)}\,(2u)^{-k^{2}-k\kappa}\,F_{k+\kappa,k}(2u)\left(1+o(1)\right), \qquad N \to \infty, \label{tcueresult}
\end{equation}
where $F_{k+\kappa,k}(x)$, as in \eqref{LUEprob}, is the distribution function of the largest eigenvalue in the $k \times k$ LUE with $k+\kappa$ degrees of freedom,
\begin{equation}
F_{k+\kappa,k}(x) = \mathbb{P}\left(\lambda^{(\mathrm{LUE}_{k,k+\kappa})}_{\mathrm{max}} < x\right). \label{LUEtruncatedCUE}
\end{equation}
\end{theorem}
\begin{proof}
See Section \ref{se:tcue}.
\end{proof}
\begin{remark}
As in Theorems \ref{th:twocharge} and \ref{th:onecharge}, the quantity \eqref{LUEtruncatedCUE} is expressible in terms of Painlev\'e transcendents. The results of \cite{TW94} imply that \eqref{LUEtruncatedCUE} is given by
\begin{equation}
F_{k+\kappa,k}(x) = \mathrm{exp}\left(-\int_{x}^{\infty}\frac{\sigma^{(\mathrm{V})}_{k,\kappa}(t)}{t}\,dt\right)
\end{equation}
where $\sigma^{(\mathrm{V})}_{k,\kappa}(t)$ satisfies the $\sigma$-form of Painlev\'e V \eqref{pvintro} with parameter $\alpha=\kappa$. Via this representation, our expansion \eqref{tcueresult} can be interpreted for both non-integer $k$ and non-integer $\kappa$.
\end{remark}
\begin{figure}
\label{fig_tCUE1}
\begin{center}
\begin{overpic}[scale=.5]{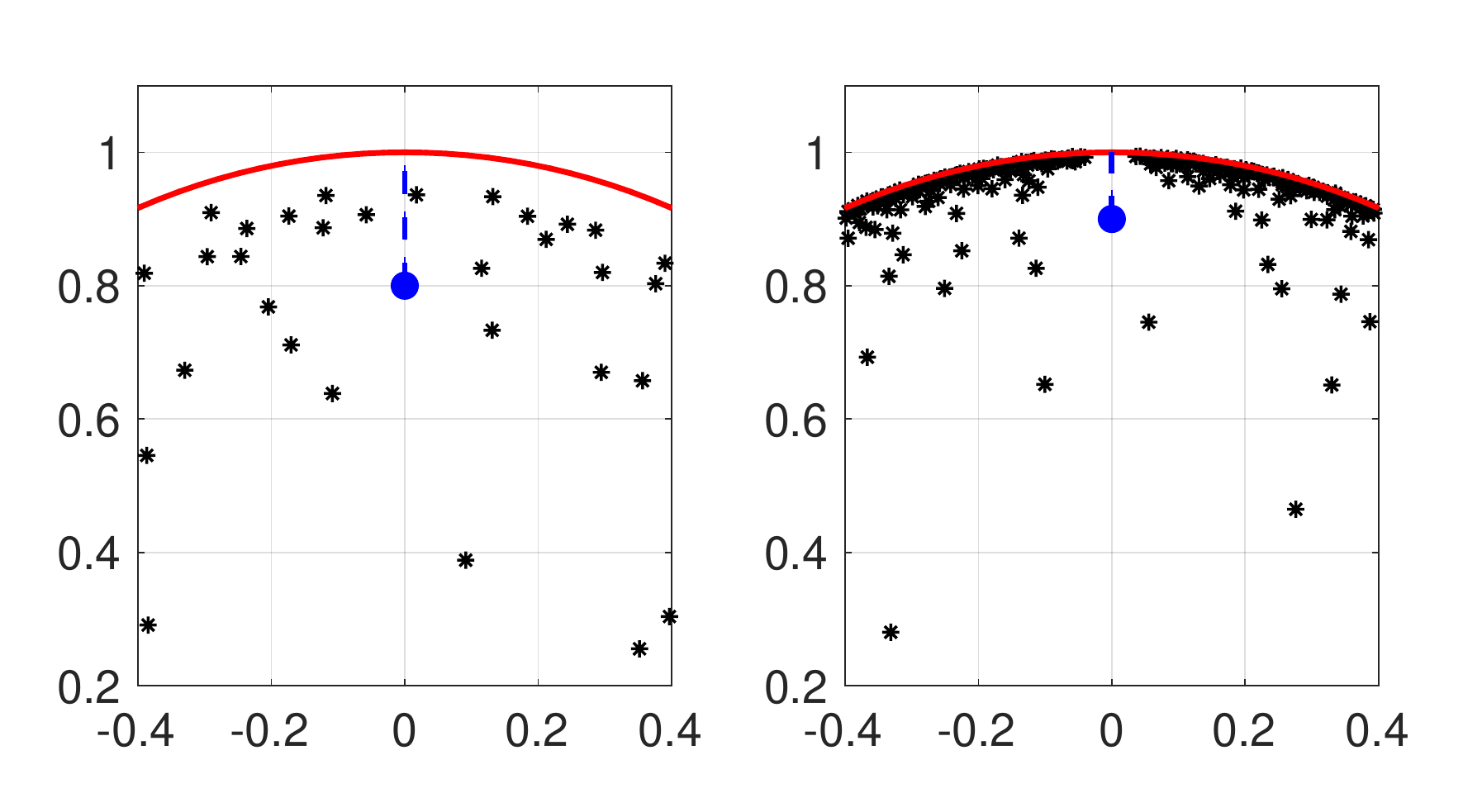}
\put(29,46){\color{blue}\small{$\mathcal{O}(N^{-1})$}}
\put(29,34){\color{blue}{$z$}}
\put(77,46){\color{blue}\small{$\mathcal{O}(N^{-1})$}}
\put(77,38){\color{blue}{$z$}}
\end{overpic}
\caption{Schematic insertion of a charge $z$ near the boundary in the truncated CUE ensemble. On the left $M=200$ and $N=180$, and on the right $M=2000$ and $N=1980$. The unit circle is depicted in red.}
\end{center}
\end{figure}
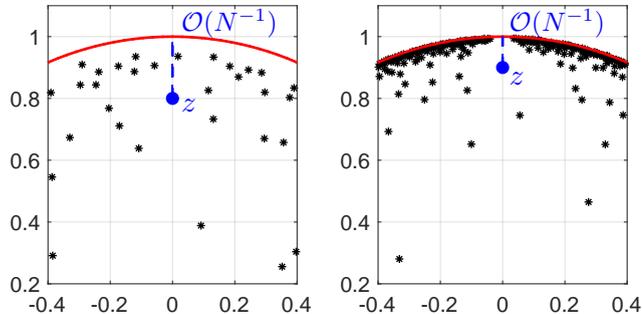

Compared to Theorem \ref{th:onecharge}, an interesting point of departure concerns the change from Painlev\'e IV to Painlev\'e V appearing in the constant term of the asymptotic expansion \eqref{tcueresult}. One notable difference here is that the boundary forms a \textit{hard edge}: eigenvalues of $T$ are forbidden from leaving the unit disc and thus a different boundary behaviour can be expected. Properties of two-dimensional models of type \eqref{law} with a hard edge have been studied recently \cite{AKM19} and were shown to deviate from the boundary statistics known for the Ginibre ensemble.

For our final result on the truncated unitary ensemble, we state an analogue of \eqref{tcueresult} at finite $N$. In order to state the result, let $W_1$ and $W_2$ be two independent LUE matrices with $p_1$ and $p_2$ degrees of freedom, both of dimension $k$. The ratio of LUE matrices
\begin{equation}
J = W_{1}(W_{1}+W_{2})^{-1}
\end{equation}
has eigenvalues supported on the interval $[0,1]$. This set of eigenvalues is known as the \textit{Jacobi Unitary Ensemble}, again see \cite{Forrester_loggas} for further background. We denote the largest eigenvalue of $J$ by $\lambda^{(\mathrm{JUE}_{k,\alpha,\beta})}_{\mathrm{max}}$ where $\alpha=p_{1}-k$ and $\beta = p_{2}-k$ is standard notation for the parameters of the ensemble.
\begin{theorem}[Truncated CUE and Painlev\'e VI]
\label{th:tcuefiniten}
Let $k \in \mathbb{N}$ and $|z|<1$. Then we have the exact identity
\begin{equation}
\mathbb{E}\left(|\det(T-z)|^{2k}\right) = \frac{C_{M,N,k}}{(1-|z|^{2})^{k(M-N)+k^{2}}}\mathbb{P}\left(\lambda^{(\mathrm{JUE}_{k,M-N,N})}_{\mathrm{max}} < 1-|z|^{2}\right), \label{JUEprobab}
\end{equation}
where
\begin{equation}
C_{M,N,k} := \mathbb{E}\left(|\det(T)|^{2k}\right) = \prod_{j=0}^{k-1}\frac{\Gamma(M-N+1+j)\Gamma(N+1+j)}{\Gamma(M+1+j)\Gamma(1+j)}. \label{r2kzero}
\end{equation}
\end{theorem}
\begin{proof}
See Section \ref{se:tcue}.
\end{proof}
\begin{remark}
Although \eqref{JUEprobab} holds for $|z| < 1$, the general case where $z \in \mathbb{C}$ is covered in Proposition \ref{prop:JUEduality}. A result for non-integer $k$ involving $U(N)$ group integrals is described in Theorem \ref{prop:tPVI}.
\end{remark}
\begin{remark}
As with the previous cases involving GUE and LUE, one also anticipates a Painlev\'e representation for the JUE probability on the right-hand side of \eqref{JUEprobab}. A third order equation was obtained in the same paper of Tracy and Widom \cite{TW94}, but this was not expressed in second order form and identified in terms of Painlev\'e transcendents until later work \cite{LJP99, WFC00,AvM01, BD02, FWPVI}. These studies show, with a variety of approaches, that for $0 \leq x \leq 1$,
\begin{equation}
\mathbb{P}\left(\lambda^{(\mathrm{JUE}_{k,\alpha,\beta})}_{\mathrm{max}}< x\right) = \mathrm{exp}\left(-\int_{x}^{1}\frac{\sigma^{(\mathrm{VI})}_{\alpha,\beta}(t)-b_{1}b_{2}t+\frac{b_{1}b_{2}+b_{3}b_{4}}{2}}{t(1-t)}\,dt\right), \label{pvirep}
\end{equation}
where $\sigma^{(\mathrm{VI})}_{\alpha,\beta}(t)$ satisfies the Jimbo-Miwa-Okamoto $\sigma$-form of the Painlev\'e VI equation
\begin{equation}
\sigma'(t(1-t)\sigma'')^{2}-\left(\sigma'(2\sigma+(1-2t)\sigma')+\prod_{i=1}^{4}b_{i}\right)^{2} + \prod_{i=1}^{4}(\sigma'-b_{i}^{2}) = 0, \label{pvieqn}
\end{equation}
with parameters
\begin{equation}
b_{1} = b_{2} = k+\frac{\alpha+\beta}{2}, \quad b_{3} = \frac{\alpha+\beta}{2}, \quad b_{4} = \frac{\beta-\alpha}{2}. \label{paramsb}
\end{equation}
This continues to give an interpretation of the right-hand side of \eqref{JUEprobab} even when $k$ is not an integer. We will prove in Corollary \ref{cor:consist} that the extrapolation to non-integer $k$ in this fashion is the correct interpretation.
\end{remark}
\begin{remark}
In the limit $\kappa \to 0$ the matrix $T$ reduces to a Haar distributed unitary matrix whose spectrum lies precisely on the unit circle. In that case the asymptotics of order $N^{k^{2}}$ is close to the Keating and Snaith result used to conjecture the corresponding moments of the Riemann zeta function \cite{KS00}. In this purely unitary case $\kappa=0$ the association with Painlev\'e VI and the boundary scaling in terms of Painlev\'e V was discussed in \cite{FWPVI}. More recently, Painlev\'e V arose for boundary asymptotics of \textit{derivatives} of CUE characteristic polynomials \cite{MixedMoms19}.
\end{remark}
\subsection{Multiple charges}

In this section we state our results for the correlation function \eqref{correlationfn} for $m>2$. At the edge, we found a natural generalization of the GUE type Theorem \ref{th:onecharge} which replaces the probability \eqref{boundaryerfc} with an object related to non-intersecting Brownian motion.

Recall that the transition density of a standard Brownian motion on $\mathbb{R}$, leaving from $u$ and arriving at $v$ is
\begin{equation}
p_{t}(u,v) = \frac{1}{\sqrt{2\pi t}}\,e^{-(u-v)^{2}/(2t)}.\label{bm}
\end{equation}
Now consider $k$ non-intersecting Brownian motions on $\mathbb{R}$, leaving from positions $(u_1,\ldots,u_k)$ at time $0$ and arriving at $(v_1,\ldots,v_k)$ at time $1$. The probability that all particles belong to the set $E$ at time $t$ is, up to a normalizing factor, given by the Karlin-McGregor formula (see \cite{KM59} or \cite[Sec. 7.4]{AvMV09})
\begin{equation}
\mathcal{Z}_{t}(\vec{u},\vec{v},E) = \int_{E^{k}}\,\det\bigg\{p_{t}(u_i,s_j)\bigg\}_{i,j=1}^{k}\,\det\bigg\{p_{1-t}(s_i,v_j)\bigg\}_{i,j=1}^{k}\,d\vec{s}.
\end{equation}
Of particular interest here will be the case $E = \mathbb{R}_{+}$ which corresponds to the probability that the smallest particle is positive at time $t$.
\begin{theorem}
\label{th:multiplebdry}
Consider the boundary scaling
\begin{equation}
x_{j} = z - \frac{u_j}{\overline{z}\sqrt{N}}, \qquad \overline{y_{j}} = \overline{z} - \frac{\overline{v_{j}}}{z\sqrt{N}}, \qquad j=1,\ldots,k, \label{edgescale}
\end{equation}
where $z$ is a fixed point on the boundary, $|z|=1$, with fixed complex vectors $\vec{u}, \vec{v} \in \mathbb{C}^{k}$, which may contain degeneracies. Then we have the following asymptotic expansion,
\begin{equation}
\begin{split}
\mathbb{E}&\left(\prod_{j=1}^{k}\det(G_{N}-x_j)\det(G_{N}^{\dagger}-\overline{y_j})\right) = \mathrm{exp}\left(-\sum_{j=1}^{k}\left(\sqrt{N}(u_{j}+\overline{v_{j}})-\frac{u_{j}^{2}+\overline{v_{j}}^{2}}{2}\right)+\frac{k^{2}}{2}\log(N)\right)\\
&\times \frac{(2\pi)^{k}}{k!}\,F^{(\mathrm{edge})}_{k}(\vec{u},\vec{v})\,\left(1+o(1)\right), \qquad N \to \infty, \label{multiplecharge}
\end{split}
\end{equation}
where 
\begin{equation}
F^{(\mathrm{edge})}_{k}(\vec{u},\vec{v}) = \frac{\mathcal{Z}_{1/2}(\vec{u},\overline{\vec{v}},\mathbb{R}_{+})}{\prod_{1 \leq i < j \leq k}(u_{j}-u_{i})(\overline{v_j}-\overline{v_i})}. \label{ratioedge}
\end{equation}
In the event that the vectors $\vec{u}$ or $\vec{v}$ contain degeneracies, the expansion \eqref{multiplecharge} continues to hold by taking appropriate limits in \eqref{ratioedge}.
\end{theorem}
\begin{proof}
See Section \ref{se:boundarymerge}.
\end{proof}
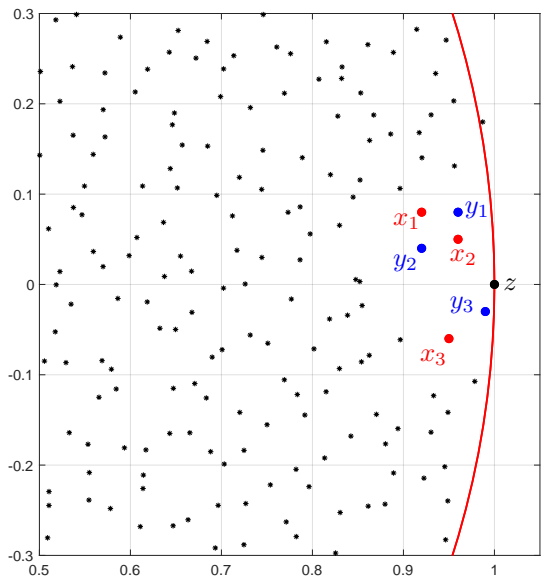
\begin{figure}
\label{fig_Ginibre3}
\begin{center}
\begin{overpic}[width=.5\textwidth]{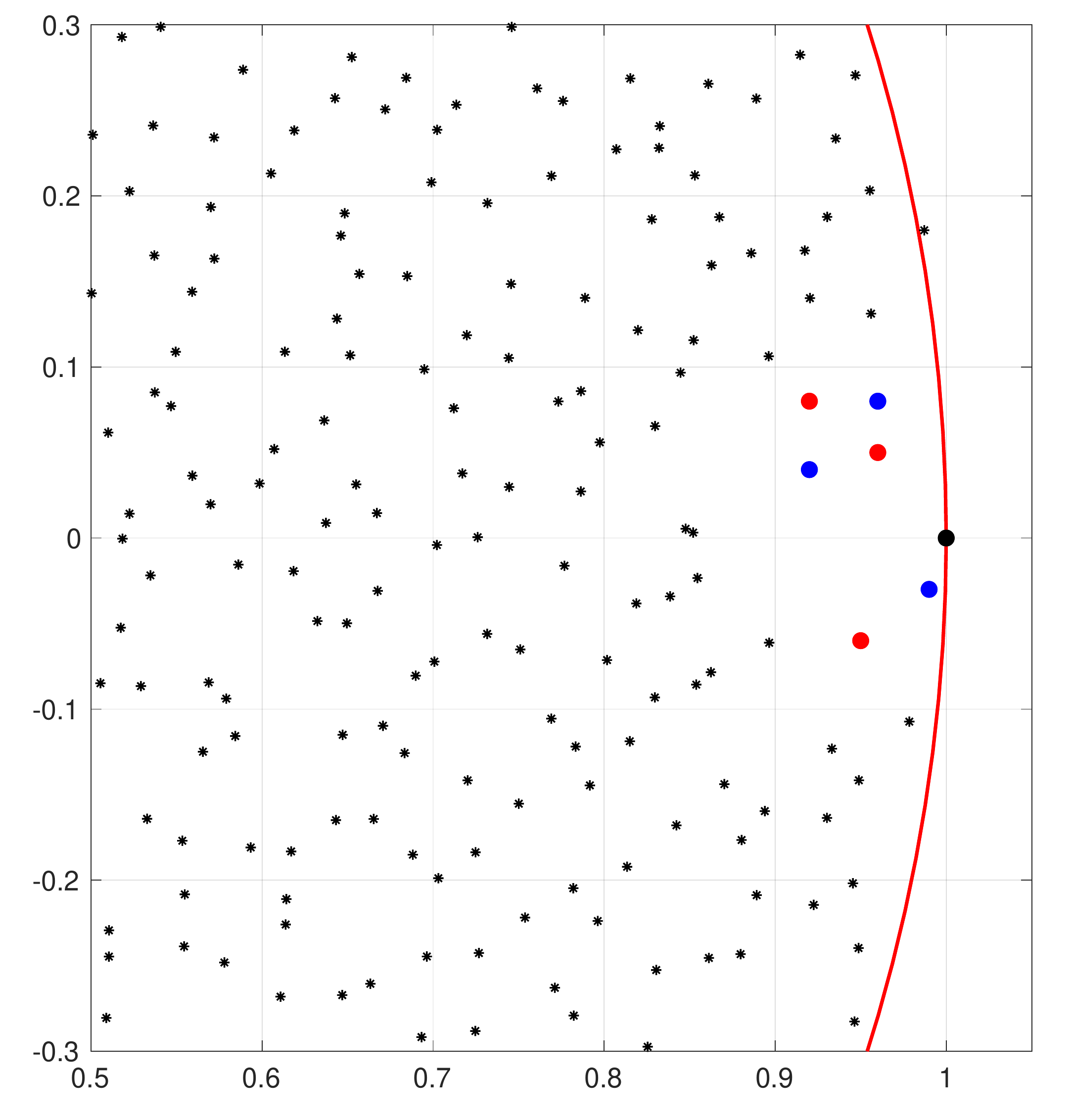}
\put(67.5,62){\color{red}\small{$x_1$}}
\put(79.5,64){\color{blue}\small{$y_1$}}
\put(77,56){\color{red}\small{$x_2$}}
\put(67.5,55){\color{blue}\small{$y_2$}}
\put(72,39){\color{red}\small{$x_3$}}
\put(77,48){\color{blue}\small{$y_3$}}
\put(86,51){\color{black}\small{$z$}}
\end{overpic}
\caption{Schematic insertion of two families of charges (with $k=3$) near the point $z=1$ on the boundary in the complex Ginibre ensemble.}
\end{center}
\end{figure}
To obtain \eqref{correlationfn}, we have to merge entries of the vectors $\vec{u}$ and $\vec{v}$ into several distinct families and eventually take $\vec{u} = \vec{v}$ in the ratio \eqref{ratioedge}. In the Brownian motion interpretation, the confluence of terminal and initial points is well studied and has appeared several times in the literature. For example, as in \cite[Section 7.4]{AvMV09} we can consider $m$ families $u^{(1)}_{i},\ldots,u^{(k_i)}_{i}$ and $v^{(1)}_{i},\ldots,v^{(k_i)}_{i}$ for $i=1,\ldots,m$ each of which are merged to a single point $a_{i}$ and $b_{i}$ respectively for each $i=1,\ldots,m$ with $k_{1}+\ldots+k_{m} = k$, finally setting $a_{i}=b_{i}$. Thus we are considering $k$ non-intersecting Brownian motions with $k_{i}$ paths starting at $a_{i}$ and terminating at $b_{i}$, leading to
\begin{equation}
\begin{split}
&\mathcal{Z}_{t}(\vec{u},\vec{v},E) \sim c_{k,t}(\vec{u},\vec{v})\int_{E^{k}}\prod_{j=1}^{k}e^{-\frac{s_{j}^{2}}{2t(1-t)}}\,\\
&\times \det\begin{pmatrix} \bigg\{s_{j}^{i-1}e^{\frac{a_{1}s_{j}}{t}}\bigg\}_{i=1,j=1}^{k_1,k}\\
\vdots\\
\bigg\{s_{j}^{i-1}e^{\frac{a_{m}s_{j}}{t}}\bigg\}_{i=1,j=1}^{k_m,k}
\end{pmatrix} \det\begin{pmatrix} \bigg\{s_{j}^{i-1}e^{\frac{b_{1}s_{j}}{1-t}}\bigg\}_{i=1,j=1}^{k_1,k}\\
\vdots\\
\bigg\{s_{j}^{i-1}e^{\frac{b_{m}s_{j}}{1-t}}\bigg\}_{i=1,j=1}^{k_m,k}
\end{pmatrix}d\vec{s}, \label{confluentboundary}
\end{split}
\end{equation}
where $c_{k,t}(\vec{u},\vec{v})$ is an explicit and simple to calculate pre-factor that we omit from the presentation. Although the vectors $\vec{u}, \vec{v}$ should be real for this probabilistic interpretation, the formulae \eqref{ratioedge} and \eqref{confluentboundary} continue to make sense for complex values. This block structure has appeared in the context of multiple Hermite polynomials and the $m$-component KP hierarchy \cite{DK04,DK07, AvMV09}. The latter in principle yields a system of PDEs generalizing the Painlev\'e characterization \eqref{p4} to the multiple charge context, thus giving a possible interpretation of the boundary asymptotics of \eqref{correlationfn} for non-integer $\{k_{i}\}_{i=1}^{m}$. The recent work \cite{LY19} studied the relation between the planar orthogonal polynomials associated with \eqref{correlationfn} with non-integer $\{k_{i}\}_{i=1}^{m}$ and type II multiple orthogonal polynomials, but only at the finite $N$ level. It would be interesting to understand if the multiple orthogonality known to be associated with \eqref{confluentboundary} is related. 

Finally we discuss the generalization of Theorem \ref{th:twocharge} to multiple charges colliding in the bulk. 
\begin{theorem}
\label{th:hcizbulk}
Fix $\delta>0$ and take $z$ inside the unit disc $|z| < 1-\delta$. Consider points $\{x_{i}\}_{i=1}^{k}$ and $\{y_{i}\}_{i=1}^{k}$ centered and scaled 
\begin{equation}
x_{i} = z+\frac{u_{i}}{\sqrt{N}}, \qquad \overline{y_{i}} = \overline{z}+\frac{\overline{v_{i}}}{\sqrt{N}}, \label{bulkscale}
\end{equation}
for $i=1,\ldots,k$. Then the following holds uniformly for $\{u_{i},v_{i}\}_{i=1}^{k}$ varying in compact subsets of $\mathbb{C}$,
\begin{equation}
\begin{split}
&\mathbb{E}\left(\prod_{i=1}^{k}\det(G_{N}-x_{i})\det(G_{N}^{\dagger}-\overline{y_{i}})\right) = \mathrm{exp}\left(Nk(|z|^{2}-1)+\sqrt{N}\sum_{i=1}^{k}(z\overline{v_{i}}+\overline{z}u_{i})+\frac{k^{2}}{2}\log(N)\right)\\
& \times \frac{(2\pi)^{k/2}}{G(1+k)}\,F^{(\mathrm{bulk})}_{k}(\vec{u},\vec{v})\,\left(1+o(1)\right), \qquad N \to \infty,\label{multicharge}
\end{split}
\end{equation}
where $F^{(\mathrm{bulk})}_{k}(\vec{u},\vec{v})$ is the following integral over the group $U(k)$ of $k \times k$ unitary matrices,
\begin{equation}
F^{(\mathrm{bulk})}_{k}(\vec{u},\vec{v}) = \int_{U(k)}e^{\mathrm{Tr}(UAU^{\dagger}\overline{B})}\,d\mu(U). \label{fkbulk}
\end{equation}
Here $d\mu(U)$ is the normalized Haar measure on $U(k)$ and
\begin{equation}
A = \mathrm{diag}(u_1,\ldots,u_k), \qquad \overline{B} = \mathrm{diag}(\overline{v_1},\ldots,\overline{v_k}).
\end{equation}
\end{theorem}
\begin{proof}
See Section \ref{se:bulk}.
\end{proof}
The integral appearing in \eqref{fkbulk} is well known in random matrix theory under the name \textit{Harish-Chandra Itzykson Zuber integral}, see \cite{IZ80}. As we discuss in Section \ref{se:bulk} there is an explicit formula for this integral. However, in order to obtain \eqref{correlationfn} or the result in Theorem \ref{th:twocharge}, a rather delicate merging procedure is required and the explicit formula is not always easy to use in this case. We explain how to overcome this in Section \ref{se:bulk}.
\begin{remark}
\label{rem:universality}
A remark about universality: Our proof of the aforementioned results relies essentially only on local asymptotics of the relevant correlation kernel, either close to the boundary or in the bulk (in contrast to the approaches in \cite{NK02,Kan05} which exploit the Gaussian structure of the Ginibre ensemble). For the more general model \eqref{law}, the correlations in the bulk and at the edge are now known to be strongly independent of the potential $V(\lambda)$ (universal) up to some restrictions on the regularity of $V(\lambda)$ and the droplet. This has been accomplished at the edge in the recent development \cite{HW17} and for the bulk in \cite{AKS18}. Indeed, considering Theorem \ref{th:onecharge} for example, based on the results in \cite{HW17} we expect that the same Painlev\'e IV object describes the collision of a point charge with the boundary of the droplet for a broad class of normal matrix models. In a different direction, bulk universality for correlations of non-Hermitian characteristic polynomials was recently investigated in \cite{Af19} where a formula essentially equivalent to \eqref{multicharge} was derived for a class of independent entry random matrices.
\end{remark}
To summarize, we tabulate below the results described in this section. Each Theorem can be viewed as a certain duality mapping the planar problem under consideration to objects typically encountered in the study of Hermitian or unitary matrices. The last column indicates the associated Painlev\'e system.\\\\
\begin{center}
{\renewcommand{\arraystretch}{1.3}\begin{tabular}{ |c||c|c|c|c| }
 \hline
 \multicolumn{5}{|c|}{Microscopic behaviour of the correlation function \eqref{correlationfn}} \\
 \hline
 Planar model & $m$ & Regime &Duality & Painlev\'e\\
 \hline
 $\mathbb{C}$-Ginibre  & 1    &Finite $N$& LUE & V\\
---''--- &   1  & Edge, $N \to \infty$ &GUE& IV\\
---''--- &   2  & Bulk, $N \to \infty$ & LUE& V\\
---''---  &  $m > 2$  & Edge, $N \to \infty$ & \footnotesize{Non-intersecting paths} & \\
---''---  &  $m > 2$  & Bulk, $N \to \infty$ & \footnotesize{Itzykson Zuber integral} & \\
T-CUE& 1 &Finite $N$ & JUE  & VI\\
---''---& 1 & Edge, $N \to \infty$ & LUE & V\\
 \hline
\end{tabular}}
\vspace{6pt}
\captionof{table}{The above table summarises the main results of the present work.}
\end{center}
The structure of this paper is as follows. In Section \ref{se:opav} we recall some well known facts about the normal matrix model \eqref{law} and its relation to determinantal structures and orthogonal polynomials in the complex plane. Then in Section \ref{se:boundary} we compute the $1$-point correlation (object \eqref{correlationfn} with $m=1$) at finite $N$ and discuss its relation to Painlev\'e V for finite $N$ and Painlev\'e IV near the boundary, proving Theorem \ref{th:onecharge}. These boundary asymptotics are then used to give a conjecture for the critical asymptotics of the partition function \eqref{generalpf} with potential \eqref{lem}. Section \ref{se:bulk} is devoted to proving the results about bulk asymptotics and merging of singularities namely Theorems \ref{th:twocharge}, \ref{th:nonintegerbulk} and \ref{th:hcizbulk}. Finally Section \ref{se:tcue} is devoted to studying the $1$-point correlation for the truncated unitary ensemble, its relation to Painlev\'e VI and boundary asymptotics in terms of Painlev\'e V.
\newpage
\section{Planar orthogonal polynomials and determinantal structures}
\label{se:opav}
In this section we set some basic notation for the paper and recall some known facts regarding the normal matrix model and planar orthogonal polynomials, see \textit{e.g.} \cite{AV03, AHM11, BGM17} and references therein for relevant background. Consider the normal matrix model corresponding to a weight $w(\lambda)$ as defined in \eqref{law}, where $w(\lambda) = e^{-NV(\lambda)} \geq 0$. From the weight $w(\lambda)$ we may construct a unique family of degree $k$ monic polynomials $p_{k}(\lambda), k=0,1,2,\ldots$ satisfying the planar orthogonality
\begin{equation}
\int_{\mathbb{C}}p_{k}(\lambda)\overline{p_{j}(\lambda)}\,w(\lambda)\,d^{2}\lambda = \delta_{k,j}h_{k} \label{orthog}
\end{equation}
and these polynomials can be used to characterize statistical properties of \eqref{law}. A distinguished role is played by a function of two variables known as the \textit{correlation kernel} of the model, defined by
\begin{equation}
K_{N}(\lambda_1,\overline{\lambda_2}) = \sqrt{w(\lambda_1)w(\lambda_2)}\sum_{j=0}^{N-1}\frac{p_{j}(\lambda_1)\overline{p_{j}(\lambda_2)}}{h_j}. \label{kernel}
\end{equation}
For example, the $k$-point eigenvalue correlation functions corresponding to \eqref{law} are given explicitly in terms of \eqref{kernel},
\begin{equation}
\begin{split}
R_{k}(\lambda_1,\ldots,\lambda_k) &:= \frac{N!}{(N-k)!}\int_{\mathbb{C}^{N-k}}P(\lambda_1,\ldots,\lambda_N)\,d^{2}\lambda_{k+1}\ldots d^{2}\lambda_{N}\\
&=\det\bigg\{K_{N}(\lambda_i,\overline{\lambda_j})\bigg\}_{i,j=1}^{k},
\end{split}
\end{equation}
where $P(\lambda_1,\ldots,\lambda_N)$ is a joint probability density function of the type \eqref{law} or \eqref{zspdf}. This identity expresses the fact that eigenvalues of normal random matrices form a determinantal point process with kernel \eqref{kernel}. The partition function \eqref{generalpf} is also expressed in terms of orthogonal polynomials, via the norming constants in \eqref{orthog} we have
\begin{equation}
\mathcal{Z}_{N} = N!\det\bigg\{\int_{\mathbb{C}}\lambda^{i}\overline{\lambda}^{j}w(\lambda)\,d^{2}\lambda \bigg\}_{i,j=0}^{N-1} = N!\prod_{j=0}^{N-1}h_{j}. \label{norms}
\end{equation}
A simple but very useful fact we will exploit throughout the paper is that if $w(\lambda)$ is rotationally invariant, i.e. $w(\lambda) = w(|\lambda|)$, then the planar orthogonal polynomials are the monomials, i.e. $p_{j}(\lambda) = \lambda^{j}$ for each non-negative integer $j$. 

The kernel \eqref{kernel} is also relevant for the computation of our central object of interest \eqref{correlationfn}, as implied by the following result. In what follows, and in the rest of the paper, we will use the notation
\begin{equation}
\Delta(\vec{x}) := \prod_{1 \leq i < j \leq k}(x_{j}-x_{i}) = \det\bigg\{x_{i}^{j-1}\bigg\}_{i,j=1}^{k}
\end{equation}
to denote the Vandermonde determinant in the entries of the vector $\vec{x}$, whose dimension may vary depending on the context.
\begin{theorem}[Akemann and Vernizzi \cite{AV03}]
Consider the \textit{polynomial kernel}
\begin{equation}
B_{N}(x,\overline{y}) := \sum_{j=0}^{N-1}\frac{p_{j}(x)\overline{p_{j}(y)}}{h_j} = K_{N}(x,\overline{y})[w(x)w(y)]^{-1/2} \label{polykern}
\end{equation}
and generic complex numbers $\{x_{j}\}_{j=1}^{k}$ and $\{y_{j}\}_{j=1}^{k}$. The following exact identity holds
\begin{equation}
\begin{split}
\mathbb{E}\left(\prod_{j=1}^{N}\prod_{i=1}^{k}(\lambda_{j}-x_i)(\overline{\lambda_{j}}-\overline{y_i})\right) &= \frac{\det\bigg\{B_{N+k}(x_i,\overline{y_j})\bigg\}_{i,j=1}^{k}}{\Delta(\vec{x})\overline{\Delta(\vec{y})}}\,\prod_{j=0}^{k-1}h_{j+N} \label{akevern}
\end{split}
\end{equation}
where the expectation is taken with respect to \eqref{law}.
\end{theorem}
The idea behind this result goes back to Br\'ezin and Hikami \cite{BH00} who developed it in the context of Hermitian random matrices. In order to reproduce averages of the type \eqref{correlationfn}, we have to merge some of the points $\{x_{j}\}_{j=1}^{k}$, $\{y_{j}\}_{j=1}^{k}$ which is easily done with L'H\^{o}pital's rule and basic facts for differentiating determinants. For example, in the case $m=1$ of \eqref{correlationfn}, we have to merge together all points $x_{i}\to x$ and $y_{i} \to y$ for all $i=1,\ldots,k$ which gives 
\begin{equation}
\frac{\det\bigg\{A(x_i,y_j)\bigg\}_{i,j=1}^{k}}{\Delta(\vec{x})\Delta(\vec{y})} \to \left(\prod_{j=0}^{k-1}\frac{1}{(j!)^{2}}\right)\,\det\bigg\{\frac{\partial^{i+j-2}A(x,y)}{\partial x^{i-1}\,\partial y^{j-1}}\bigg\}_{i,j=1}^{k} \label{mergeakevern}
\end{equation}
for any smooth function of two variables $A(x,y)$. This representation also turns out to be useful for studying asymptotics, as it reduces the analysis to that of the two-variable kernel \eqref{polykern} and its derivatives. In order to calculate (or identify) the determinant on the right-hand side of \eqref{mergeakevern} we will often exploit the following.
\begin{lemma}[Andr\'eief identity]
\label{lem:andre}
For a domain $D$ and two sets of integrable functions $\{f_{j}(t)\}_{j=1}^{k}$ and $\{g_{j}(t)\}_{j=1}^{k}$ we have
\begin{equation}
\det\bigg\{\int_{D}f_{i}(t)\,g_{j}(t)\,dt\bigg\}_{i,j=1}^{k} = \frac{1}{k!}\int_{D^{k}}\,\det\{f_{j}(t_i)\}_{i,j=1}^{k}\,\det\{g_{j}(t_i)\}_{i,j=1}^{k}\,d\vec{t}.
\end{equation}
\end{lemma}


\section{The Ginibre case: finite $N$ and boundary asymptotics}
\label{se:boundary}
In this section we begin by discussing in some detail the object \eqref{correlationfn} with $m=1$, that is the moments
\begin{equation}
R_{2k}(z) = \mathbb{E}\left(|\det(G_{N}-z)|^{2k}\right), \label{Robject}
\end{equation}
where $G_{N}$ is a standard complex Ginibre random matrix. Of interest will be the relation to Painlev\'e transcendents at finite $N$ and the asymptotics $N \to \infty$. In particular we will prove Theorem \ref{th:onecharge} and discuss the problem of continuing $k$ off the integers. Then we apply the obtained results to the lemniscate partition function. Finally, we extend the considerations to multiple products of characteristic polynomials and prove Theorem \ref{th:multiplebdry}.
\subsection{Duality with the LUE}
\label{se:dualitylue}
Recall the definition of the Laguerre Unitary Ensemble from the introduction: form the matrix $W=G_{k,p}G_{k,p}^{\dagger}$ where $G_{k,p}$ is a $k \times p$ matrix of \textit{i.i.d.} standard complex Gaussian random variables and we assume $p \geq k$. The eigenvalues $t_{1},\ldots,t_{k}$ of $W$ are all non-negative and have the following well-known joint probability density function
\begin{equation}
P(t_{1},\ldots,t_{k}) = \frac{1}{C^{(\mathrm{LUE}_{\alpha})}_{k}}\prod_{j=1}^{k}t_{j}^{\alpha}e^{-t_{j}}\Delta^{2}(\vec{t}),
\end{equation}
where the normalization constant is given explicitly by
\begin{equation}
C^{(\mathrm{LUE}_{\alpha})}_{k} := \int_{[0,\infty)^{k}}\prod_{j=1}^{k}t_{j}^{\alpha}e^{-t_{j}}\,\Delta^{2}(\vec{t})\,d\vec{t} = \prod_{j=0}^{k-1}\Gamma(\alpha+j+1)\Gamma(j+2). \label{LUEnormconst}
\end{equation}
The quantity $\alpha = p-k$ is referred to as the \textit{parameter} of the LUE and $k$ as the size, see \cite{Forrester_loggas} for further details.
\begin{proposition}[\cite{NK02,FR09}]
The following exact duality identity holds:
\begin{align}
R_{2k}(z) &= \frac{N^{k^{2}/2}}{\prod_{j=0}^{k-1}j!(j+1)!}\int_{[0,\infty)^{k}}\prod_{j=1}^{k}e^{-\sqrt{N}t_{j}}\left(|z|^{2}+\frac{t_{j}}{\sqrt{N}}\right)^{N}\Delta^{2}(\vec{t})\,d\vec{t} \label{LUEduality}\\
&= N^{-Nk}e^{Nk|z|^{2}}\left(\prod_{j=1}^{k}\frac{\Gamma(j+N)}{\Gamma(j)}\right)\mathbb{P}(\lambda_{1} > N|z|^{2}) \label{smallesteig}
\end{align}
where $\lambda_{1}$ is the smallest eigenvalue in the Laguerre Unitary Ensemble with parameter $\alpha=N$ and size $k$. 
\label{prop:lue}
\end{proposition}
To our knowledge the identity \eqref{LUEduality} first appeared in \cite{NK02} where it was derived with Grassmann integration techniques and was then generalized by Forrester and Rains in \cite{FR09} using a symmetric functions approach. Similar dualities involving positive definite matrices have appeared in the work \cite{FK07}, see also \cite{FGS18} for a different generalization, though these approaches appear to be quite specific to the Ginibre setting. In the following we will give a different proof using the general formula \eqref{akevern}.
\begin{proof}[Proof of Proposition \ref{prop:lue}]
Starting with identity \eqref{akevern}, the weight in the Ginibre ensemble\footnote{Here it is convenient to begin with the $N$-independent weight, to obtain \eqref{LUEduality} the factor of $N$ can be easily restored at the end of the calculation.} is $w(\lambda) = e^{-|\lambda|^{2}}$ and the corresponding planar orthogonal polynomials are $p_{j}(\lambda) =  \lambda^{j}$. The norms follow from an explicit integration as $h_{j} = \pi j!$ and this gives rise to the usual Ginibre finite $N$ kernel
\begin{equation}
B_{N+k}(x,y) = \frac{1}{\pi}\sum_{j=0}^{N+k-1}\frac{(xy)^{j}}{j!}= \frac{x^{N+k}}{\pi\Gamma(N+k)}\int_{0}^{\infty}e^{-tx}(y+t)^{N+k-1}\,dt, \label{ginkern}
\end{equation}
where, for simplicity (and without loss of generality), we have assumed $x>0$ and $y>0$. The Ginibre kernel has been written in this way to facilitate computing its partial derivatives, as in \eqref{mergeakevern}. Since the multiplicative factors in \eqref{ginkern} drop out of the determinant, it is enough to calculate $\det\{A(x_i,y_j)\}_{i,j=1}^{k}$ where
\begin{equation}
A(x,y) = \frac{1}{\Gamma(N+k)}\int_{0}^{\infty}e^{-tx}(y+t)^{N+k-1}\,dt \label{Afunc}.
\end{equation}
Now we merge the points $x_{1},\ldots,x_{k} \to x$ and $y_{1},\ldots,y_{k} \to y$. Differentiating and applying Lemma \ref{lem:andre} yields
\begin{equation}
\begin{split}	
&\det\bigg\{\frac{\partial^{i+j}A(x,y)}{\partial x^{i} \partial y^{j}}\bigg\}_{i,j=0}^{k-1} = c_{N,k}\det\bigg\{\int_{0}^{\infty}t^{i-1}e^{-tx}(y+t)^{N+k-j}\,dt\bigg\}_{i,j=1}^{k}\\
&=\frac{c_{N,k}}{k!}\int_{[0,\infty)^{k}}\prod_{j=1}^{k}e^{-t_{j}x}(y+t_{j})^{N}\det\bigg\{t_{i}^{j-1}\bigg\}_{i,j=1}^{k}\det\bigg\{(y+t_{i})^{k-j}\bigg\}_{i,j=1}^{k}\,d\vec{t}\\
&=\frac{c_{N,k}(-1)^{\frac{k(k-1)}{2}}}{k!}\int_{[0,\infty)^{k}}\prod_{j=1}^{k}e^{-xt_{j}}(y+t_{j})^{N}\,\Delta^{2}(\vec{t})\,d\vec{t},
\end{split}
\end{equation}
where
\begin{equation}
c_{N,k} = (-1)^{\frac{k(k-1)}{2}}\prod_{j=1}^{k}\frac{1}{\Gamma(N+k-(j-1))}.
\end{equation}
Collecting the multiplicative factors in \eqref{ginkern} and the product of norms in \eqref{akevern} gives, using \eqref{mergeakevern},
\begin{equation}
\begin{split}
\mathbb{E}&\left(\det(G-x)^{k}\det(G^{\dagger}-y)^{k}\right)\\
&=\left(\prod_{j=1}^{k}\frac{1}{(j-1)!j!}\right)\int_{[0,\infty)^{k}}\prod_{j=1}^{k}e^{-t_{j}}(xy+t_{j})^{N}\Delta^{2}(\vec{t})\, d\vec{t}, \label{finGineq}
\end{split}
\end{equation}
where we made the change of variables $t_{j} \to t_{j}/x$. As an identity between polynomials in $x$ and $y$, \eqref{finGineq} now holds for any $x,y \in \mathbb{C}$. Identity \eqref{LUEduality} follows from the rescalings $x \to \sqrt{N} x$, $y \to \sqrt{N} y$ and $t_{j} \to \sqrt{N}t_{j}$.
\end{proof}

Now we will prove Theorem \ref{th:onecharge}. 
\begin{proof}[Proof of Theorem \ref{th:onecharge}]
As our starting point we consider the duality formula \eqref{LUEduality}. Changing variable $t_{j} \to t_{j}+\sqrt{N}(1-|z|^{2})$ for each $j=1,\ldots,k$ transforms the integral to
\begin{equation}
e^{Nk(|z|^{2}-1)}\int_{(-\sqrt{N}(1-|z|^{2}),\infty)^{k}}\,\prod_{j=1}^{k}\mathrm{exp}\left(-\sqrt{N}t_{j}+N\log\left(1+t_{j}/\sqrt{N}\right)\right)\,\Delta^{2}(\vec{t})\,d\vec{t}.
\end{equation}
Now notice that pointwise, we have the limit
\begin{equation}
\mathrm{exp}\left(-\sqrt{N}t_{j}+N\log\left(1+t_{j}/\sqrt{N}\right)\right) \to e^{-t_{j}^{2}/2}, \qquad N \to \infty. \label{pointwiselimit}
\end{equation}
The claim is that in fact the following estimate is uniform provided $z$ stays inside a disc of radius $1+c/\sqrt{N}$ for some constant $c$,
\begin{equation}
\begin{split}
\int_{(-\sqrt{N}(1-|z|^{2}),\infty)^{k}}&\,\prod_{j=1}^{k}\mathrm{exp}\left(-\sqrt{N}t_{j}+N\log\left(1+t_{j}/\sqrt{N}\right)\right)\,\Delta^{2}(\vec{t})\,d\vec{t}\\ &= \int_{(-\sqrt{N}(1-|z|^{2}),\infty)^{k}}\,\prod_{j=1}^{k}e^{-t_{j}^{2}/2}\,\Delta^{2}(\vec{t})\,d\vec{t}\,\,\times(1+o(1)), \qquad N \to \infty \label{unifapprox}\\
&= \mathcal{Z}^{(\mathrm{GUE})}_{k}\,\mathbb{P}\left(\lambda^{(\mathrm{GUE}_{k})}_{\mathrm{max}}<\sqrt{N}(1-|z|^{2})\right)\,\,\times(1+o(1)), \qquad N \to \infty.
\end{split}
\end{equation}
To pass to the final equality we simply recognised the joint density of GUE eigenvalues, up to a normalization factor given explicitly by
\begin{equation}
\mathcal{Z}^{(\mathrm{GUE})}_{k} := \int_{\mathbb{R}^{k}}\prod_{j=1}^{k}e^{-t_{j}^{2}/2}\,\Delta^{2}(\vec{t})\,d\vec{t} = (2\pi)^{k/2}\prod_{j=0}^{k-1}(j+1)!.
\end{equation}
Inserting these results into \eqref{LUEduality} proves Theorem \ref{th:onecharge} after recalling that $G(1+k) = \prod_{j=0}^{k-1}j!$. It remains to verify the uniform asymptotics in the first equality of \eqref{unifapprox}. To this end, observe that the magnitude of the resulting error in the \textit{difference} between the two integrals is largest when $z=0$ (since $z$ only enters through the region of integration, and this is maximized when $z=0$). The condition $|z| \leq 1+c/\sqrt{N}$ ensures that the lower limit of integration remains bounded from above and permits the multiplicative error bound in \eqref{unifapprox}. Hence the verification of \eqref{unifapprox} reduces to showing that
\begin{equation}
\begin{split}
&\lim_{N \to \infty}\int_{(-\sqrt{N},\infty)^{k}}\,\prod_{j=1}^{k}\mathrm{exp}\left(-\sqrt{N}t_{j}+N\log\left(1+\frac{t_{j}}{\sqrt{N}}\right)\right)\,\Delta^{2}(\vec{t})\,d\vec{t}\\
&= \int_{\mathbb{R}^{k}}\prod_{j=1}^{k}e^{-t_{j}^{2}/2}\,\Delta^{2}(\vec{t})\,d\vec{t}
\end{split}
\end{equation}
and this can be proved with a dominated convergence argument, for example using the pointwise limit \eqref{pointwiselimit} and the bound
\begin{equation}
-\sqrt{N}t+N\log\left(1+\frac{t}{\sqrt{N}}\right) =-\int_{0}^{t}\frac{s\,ds}{1+\frac{s}{\sqrt{N}}} \leq -\frac{t^{2}}{2(1+|t|)}, \qquad t > -\sqrt{N}.
\end{equation}
\end{proof}
\begin{remark}
The asymptotic expansion in Theorem \ref{th:onecharge} takes a different form if $z$ is strictly outside the unit disc, see Appendix \ref{app:gff}. It is natural to ask whether one can write down expansions which combine these two in a uniform way. When $k=1$ the right-hand side of \eqref{LUEduality} is an incomplete Gamma function and the uniform transition asymptotics are well studied, see \textit{e.g.} \cite{BW07}. However, these make use of a (contour) integral representation quite different from the right-hand side of \eqref{LUEduality} and as our concern here is more the relation to Painlev\'e, we leave this question for future work.
\end{remark}
\subsection{Reduction to integrals over $U(N)$ for general $\gamma$ and Painlev\'e V}
\label{se:redgroup}
As explained in the introduction, it is interesting to consider extrapolating Theorem \ref{th:onecharge} to non-integer values of $k$. We will now consider this problem, starting with the second representation \eqref{smallesteig} at finite $N$. The work of Tracy and Widom \cite{TW94} implies the representation
\begin{equation}
\mathbb{P}(\lambda_{1} > x) = \mathrm{exp}\left(\int_{0}^{x}\frac{\sigma^{(\mathrm{V})}_{k,N}(t)}{t}\,dt\right), \label{lam1TW}
\end{equation}
where $\sigma^{(\mathrm{V})}_{k,N}(t)$ satisfies the Jimbo-Miwa-Okamoto $\sigma$-form of the Painlev\'e V equation \eqref{pvintro} with parameter $\alpha=N$. We remark that setting $z=0$ in \eqref{smallesteig} shows that the constant pre-factor is precisely $R_{2k}(0)$, so that with \eqref{lam1TW}, equation \eqref{smallesteig} becomes
\begin{equation}
R_{2k}(z) = R_{2k}(0)\,e^{N|z|^{2}k}\,\mathrm{exp}\left(\int_{0}^{N|z|^{2}}\frac{\sigma^{(\mathrm{V})}_{k,N}(t)}{t}\,dt\right). \label{r2ksecond}
\end{equation}
As $k$ now plays the role of a \textit{parameter} in the differential equation \eqref{pvintro}, in principle it can be considered for non-integer values. Starting with the left-hand side of \eqref{r2ksecond} for non-integer $k=\frac{\gamma}{2}$, we now show that the same equation \eqref{pvintro} with $k = \frac{\gamma}{2}$ gives the correct interpretation of $R_{\gamma}(z)$. Via a somewhat different route, this has also been demonstrated in the work of Kanzieper \cite{Kan05} in a different context. However, as the association of \eqref{Robject} with Painlev\'e V seems not well known, we will give an alternative proof below. The proof exploits a Toeplitz determinant identity obtained rather recently in the work \cite{WW18} where the identification in terms of Painlev\'e V did not appear. For completeness, we sketch the main ideas in \cite{WW18} which lead to the following.
\begin{theorem}
\label{th:ginibretocue}
Let $\gamma>-2$ and let $d\mu(U)$ denote the normalized Haar measure on the group of $N \times N$ unitary matrices. Then we have the identity
\begin{equation}\label{ZN_Toeplitz}
\begin{split}
R_{\gamma}(z) &= R_{\gamma}(0)\int_{U(N)}\,\mathrm{det}(U)^{-\frac{\gamma}{4}}|\mathrm{det}(I+U)|^{\frac{\gamma}{2}}\,\mathrm{exp}\left(N|z|^{2}\mathrm{Tr}(U)\right)d\mu(U)\\
&= \frac{R_{\gamma}(0)}{N!(2\pi)^{N}}\int_{[-\pi,\pi]^{N}}\,\prod_{j=1}^{N}e^{-\frac{i\gamma \theta_{j}}{4}}|1+e^{i\theta_{j}}|^{\frac{\gamma}{2}}e^{N|z|^{2}e^{i\theta_{j}}}\,|\Delta(e^{i\theta})|^{2}\,d\vec{\theta},
\end{split}
\end{equation}
where the constant is
\begin{equation}
R_{\gamma}(0) = N^{-\frac{\gamma N}{2}}\prod_{j=0}^{N-1} \frac{\Gamma(\frac{\gamma}{2}+j+1)}{\Gamma(j+1)}. \label{angam}
\end{equation}
Furthermore, we have the representation in terms of solutions of \eqref{pvintro},
\begin{equation}
R_{\gamma}(z) = R_{\gamma}(0)\,e^{N|z|^{2}\frac{\gamma}{2}}\,\mathrm{exp}\left(\int_{0}^{N|z|^{2}}\,\frac{\sigma^{(\mathrm{V})}_{\gamma/2,N}(t)}{t}\,dt\right) \label{TWformula}.
\end{equation}
\end{theorem}
\begin{proof}
Since $\gamma$ is no longer an even integer, formula \eqref{akevern} does not apply. Instead we apply identity \eqref{norms} and work with the planar polynomials $p_{j}(\lambda)$ orthogonal with respect to the weight $|\lambda-z|^{\gamma}e^{-N|\lambda|^{2}}$ with $\lambda \in \mathbb{C}$. Note that we also cannot use the previous trick to compute $p_{j}(\lambda)$ because the weight is not radial, although we can assume that $z = |z| > 0$ without loss of generality. The fundamental idea is to use Green's theorem to reduce the planar orthogonality of $p_{j}(\lambda)$ to orthogonality on a suitable contour, which can then be deformed to the unit circle. In the present context this idea goes back to \cite{BBLM15}. Indeed, we can write the orthogonality \eqref{orthog} in the differential form ($k \leq j$):
\begin{equation}
\begin{split}
h_{j}\delta_{j,k} &= \int_{\mathbb{C}}p_{j}(\lambda)(\overline{\lambda}-z)^{k}|\lambda-z|^{\gamma}e^{-N|\lambda|^{2}}\,d^{2}\lambda\\
&=\lim_{r \to \infty}\int_{|\lambda|\leq r}p_{j}(\lambda)\frac{\partial}{\partial \overline{\lambda}}\,h(\lambda,\overline{\lambda},z)\,d^{2}\lambda,
\end{split}
\end{equation}
where
\begin{equation}
h(\lambda,\overline{\lambda},z) = (\lambda-z)^{\frac{\gamma}{2}}\int_{z}^{\overline{\lambda}}(s-z)^{\frac{\gamma}{2}+k}e^{-N\lambda s}\,ds,
\end{equation}
and the roots are defined using the principal branch. Applying Green's theorem and writing $h(\lambda,\overline{\lambda},z)$ in terms of the Gamma function (see \cite[Appendix A]{WW18} for full details) yields
\begin{equation}
\begin{split}
h_{j}\delta_{j,k} = \frac{\pi \Gamma(\frac{\gamma}{2}+k+1)}{N^{\frac{\gamma}{2}+k+1}}\oint_{\Sigma}p_{j}(\lambda)\lambda^{-k-\frac{\gamma}{2}}(\lambda-z)^{\frac{\gamma}{2}}e^{-Nz\lambda}\frac{d\lambda}{2\pi i \lambda}. \label{ginibrecontour}
\end{split}
\end{equation}
The contour $\Sigma$ is any closed curve encircling the interval $[0,z]$, possibly passing through $z$, but no other points of $[0,z]$, and oriented counterclockwise; let us take $\Sigma = zS^{1}$ where $S^{1}$ is the unit circle. The change of variable $\lambda \to z\lambda$ implies that we get norms on the unit circle of the form
\begin{equation}
\begin{split}
\tilde{h}_{j}\delta_{j,k} &:= \oint_{\Sigma}p_{j}(\lambda)\lambda^{-k-\frac{\gamma}{2}}(\lambda-z)^{\frac{\gamma}{2}}e^{-Nz\lambda}\frac{d\lambda}{2\pi i \lambda}\\
&= \oint_{S^{1}}\tilde{p}_{j}(\lambda)\lambda^{-k-\frac{\gamma}{4}}|\lambda+1|^{\frac{\gamma}{2}}e^{Nz^{2}\lambda}\frac{d\lambda}{2\pi i \lambda}, \label{pjetc}
\end{split}
\end{equation}
where $\tilde{p}_{j}(\lambda)$ is another family of degree $j$ monic polynomials. To get the second equality in \eqref{pjetc} we exploited the fact that writing $\lambda=e^{i\theta}$, with $\theta\in(-\pi,\pi)$ we have the identity
\begin{equation}
\lambda^{-\frac{\gamma}{2}}(\lambda-1)^{\frac{\gamma}{2}} = e^{i\textrm{sgn}(\theta)\gamma\pi/4}\lambda^{-\frac{\gamma}{4}}|\lambda-1|^{\frac{\gamma}{2}}, \qquad \theta \in (-\pi,\pi) \label{idabs}
\end{equation}
and scaled $\lambda \to -\lambda$. Now applying \eqref{norms} in reverse we get the expression as a Toeplitz determinant
\begin{equation}
\begin{split}
R_{\gamma}(z) &= \frac{N!\prod_{j=0}^{N-1}h_{j}}{\mathcal{Z}^{(\mathrm{Gin})}_{N}} = N^{-\frac{\gamma N}{2}}\prod_{j=0}^{N-1}\frac{\Gamma(\frac{\gamma}{2}+j+1)}{\Gamma(j+1)}\prod_{j=0}^{N-1}\tilde{h}_{j}\\
&= R_{\gamma}(0)\,\det\bigg\{\oint_{S^{1}}\lambda^{j-i-\frac{\gamma}{4}}|\lambda+1|^{\frac{\gamma}{2}}e^{Nz^{2}\lambda}\frac{d\lambda}{2\pi i \lambda}\bigg\}_{i,j=0}^{N-1}. \label{tdetgin}
\end{split}
\end{equation}
Apart from the simple identity \eqref{idabs}, the determinant \eqref{tdetgin} is of the same form found in \cite[Lemma 2.5]{WW18}. The explicit form of the pre-factor $R_{\gamma}(0)$ is obtained by applying identity \eqref{norms}, using that the relevant orthogonality weight is radially symmetric. Now \eqref{ZN_Toeplitz} follows from Heine's identity, expressing the above Toeplitz determinant as an integral over the unitary group. To arrive at \eqref{TWformula} we simply observe that the quantities on the right-hand side of \eqref{ZN_Toeplitz} have known associations with Painlev\'e V, see for example \cite[Proposition 4.2]{AvMPV} and \cite[Eq. (1.46)]{FWPV} where these group integrals appear explicitly and are identified in terms of solutions of the Painlev\'e V equation.
\end{proof}
\begin{remark}
We make some observations implied by the above calculation. The polynomials appearing in \eqref{pjetc} (orthogonal on the unit circle) have been studied in relation to Riemann-Hilbert problems in the works \cite{FW04, FW06}, which involve Painlev\'e V. We can therefore make the observation that, as implied by \eqref{pjetc}, up to some explicit constants the results in \cite{FW04,FW06} also apply to the corresponding \textit{planar} orthogonal polynomials. At the level of asymptotics, the recurrence coefficients for these polynomials were studied in \cite{DK09} using Riemann-Hilbert techniques related to a Painlev\'e IV parametrix. 
\end{remark}
\begin{remark}
In the case $\gamma=2k$, the Toeplitz determinant \eqref{tdetgin} can be written as the generating function of certain combinatorial objects, associated with the longest increasing subsequence problem \cite{TWwords}. The appearance of such generating functions in the present context of the complex Ginibre ensemble seems to be a new observation. The equivalence of \eqref{ZN_Toeplitz} and \eqref{LUEduality} for $\gamma = 2k$ implies an interesting duality between averages over the unitary group and over the Laguerre ensemble of positive definite matrices. This duality appeared from a different point of view in \cite[Proposition 3.9]{FWPV}, which contains further discussion based on hypergeometric functions of matrix argument.
\end{remark}

We will now show that the appearance of Painlev\'e IV in Theorem \ref{th:onecharge} can be justified for non-integer exponents $k = \frac{\gamma}{2}$, at least to a physical level of rigour. The idea is to take the representation \eqref{TWformula} and rescale the differential equation \eqref{pvintro} near the boundary, setting $|z|^{2} = 1-\frac{2u}{\sqrt{N}}$ with $u$ fixed. First, notice that applying known asymptotics of the Barnes G-function (see \cite[Eq. 5.17.5]{NIST:DLMF}), the pre-factor gives the main terms in the asymptotic expansion of Theorem \ref{th:onecharge}:
\begin{equation}
R_{\gamma}(0) = N^{-\frac{\gamma N}{2}}\frac{G(\frac{\gamma}{2}+N+1)}{G(1+\frac{\gamma}{2})G(N+1)} = e^{-\frac{\gamma}{2}N}\,N^{\frac{\gamma^{2}}{8}}\,\frac{(2\pi)^{\frac{\gamma}{4}}}{G(1+\frac{\gamma}{2})}\,(1+o(1)), \qquad N \to \infty.
\end{equation}
Inserting this into \eqref{TWformula} and changing variables $t=N(1-s/\sqrt{N})$, we obtain
\begin{equation}
R_{\gamma}(z) = e^{\frac{N\gamma}{2}(|z|^{2}-1)}\frac{N^{\frac{\gamma^{2}}{8}}(2\pi)^{\frac{\gamma}{4}}}{G(1+\frac{\gamma}{2})}\,\mathrm{exp}\left(\int_{2u}^{\sqrt{N}}\frac{N^{-1/2}\sigma^{(\mathrm{V})}_{\gamma/2,N}(N(1-s/\sqrt{N}))}{1-s/\sqrt{N}}\,ds\right)\,(1+o(1)), \qquad N \to \infty. \label{RgamPain}
\end{equation}
Now we claim that the function $v(s) = -N^{-1/2}\sigma^{(\mathrm{V})}_{\gamma/2,N}(N(1-s/\sqrt{N}))$ has a limit. The equation satisfied by $v(s)$ is
\begin{equation}
\begin{split}
(\sqrt{N}(1-s/\sqrt{N})v'')^{2}&-[-\sqrt{N}v-N(1-s/\sqrt{N})v'+2(v')^{2}+(\gamma+N)v']^{2}\\
&+4(v')^{2}(N+\gamma/2+v')(\gamma/2+v')=0. \label{scaledv}
\end{split}
\end{equation}
Equating terms of order $N$ (leading order) in the above equation, we obtain a limiting equation
\begin{equation}
(v'')^{2}+4(v')^{2}(v'+\gamma/2)-(sv'-v)^{2}=0, \label{PIV}
\end{equation}
which is precisely the $\sigma$-form of Painlev\'e IV \eqref{p4} of Theorem \ref{th:onecharge} with $k=\frac{\gamma}{2}$. We do not attempt to justify the interchange of limits involved here, which we expect to require Riemann-Hilbert techniques, though we note that this is the only step lacking full mathematical rigour in this general exponent setting.

\subsection{Application to the lemniscate ensemble}
\label{se:lem}
The partition function of the lemniscate particle system is defined by
\begin{equation}
\mathcal{Z}^{(\mathrm{Lem}_{d})}_{N}(t) := \int_{\mathbb{C}^{N}}\prod_{j=1}^{N}e^{-N(|\lambda_{j}|^{2d}-t(\lambda_{j}^{d}+\overline{\lambda_{j}}^{d}))}\,|\Delta(\vec{\lambda})|^{2}\,d\vec{\lambda}, \qquad d \in \mathbb{N},\quad t \in \mathbb{R}. \label{lempart}
\end{equation}
In this section we will discuss this partition function for finite $N$, its relation to Painlev\'e transcendents and asymptotic expansions in the three regimes $t<t_{\mathrm{c}}$, $t=t_{\mathrm{c}}$ and $t>t_{\mathrm{c}}$ (recall that $t_{\mathrm{c}} := d^{-1/2}$). We begin, for convenience, by repeating Lemma \ref{lem:introredgin} from the introduction.
\begin{lemma}[Reduction to Ginibre]
\label{lem:ginred2}
We have the identity 
\begin{equation}
\mathcal{Z}^{(\mathrm{Lem}_{d})}_{Nd}(t) = e^{(Ntd)^{2}}c_{N,d}(\mathcal{Z}^{(\mathrm{Gin})}_{N})^{d}\prod_{\ell=0}^{d-1}\mathbb{E}\left(|\det(G_{N}-t\sqrt{d})|^{\gamma_{\ell}}\right) \label{ginred2}
\end{equation}
where
\begin{equation}
\gamma_{l} = -2\left(1-\frac{\ell+1}{d}\right), \qquad \ell=0,\ldots,d-1, \label{gamell2}
\end{equation}
and $c_{N,d} = (Nd)!d^{-N(Nd+2d+1)/2}(N!)^{-d}$. The partition function for the Ginibre ensemble (\textit{i.e.} \eqref{generalpf} with $V(\lambda)=|\lambda|^{2}$) is
\begin{equation}
\mathcal{Z}^{(\mathrm{Gin})}_{N} = \pi^{N}\frac{\prod_{k=1}^{N}k!}{N^{\frac{N(N+1)}{2}}}.
\end{equation}
\label{lem:redgin}
\end{lemma}
\begin{proof}
We give a quick proof based on ideas of \cite{BGM17}. We will abbreviate throughout
\begin{equation}
\,V^{(d)}(\lambda,t) = |\lambda|^{2d}-t(\lambda^{d}+\overline{\lambda}^{d}). \label{lempotlat}
\end{equation}
If $p_{j}(\lambda)$ are the monic planar polynomials orthogonal with respect to $e^{-N\,V^{(d)}(\lambda,t)}$, by the symmetry  $\lambda \to \lambda e^{2\pi i/d}$ there must exist a family of monic polynomials $q^{(\ell)}_{j}(\lambda)$ such that
\begin{equation}
p_{jd+\ell}(\lambda) = \lambda^{\ell}q^{(\ell)}_{j}(\lambda^{d}).
\end{equation}
The change of variables $u = \lambda^{d}$ comes with a Jacobian $d^{2}\lambda = d^{-2}|u|^{\frac{2}{d}-2}\,d^{2}u$ and we get
\begin{equation}
\begin{split}
h_{jd+\ell}\delta_{j,k} &= \int_{\mathbb{C}}|\lambda|^{2\ell}q^{(\ell)}_{k}(\lambda^{d})\overline{q^{(\ell)}_{j}(\lambda^{d})}e^{-Nd(|\lambda|^{2d}-t(\lambda^{d}+\overline{\lambda}^{d}))}\,d^{2}\lambda\\
&= d^{-1}e^{Ndt^{2}}\int_{\mathbb{C}}q^{(\ell)}_{k}(u+t)\overline{q^{(\ell)}_{j}(u+t)}|u+t|^{\gamma_{\ell}}e^{-Nd|u|^{2}}\, d^{2}u
\end{split} \label{changevars}
\end{equation}
where the $\gamma_{\ell}$ exponents are as in definition \eqref{gamell2} (note that a factor $d^{-1}$ is absorbed from the $d$-fold rotation). Scaling $u \to u/\sqrt{d}$ and using \eqref{norms} twice, we find
\begin{equation}
\mathcal{Z}^{(\mathrm{Lem}_{d})}_{Nd}(t) = (Nd)!\prod_{\ell=0}^{d-1}\prod_{j=0}^{N-1}h_{jd+\ell} = e^{(Ntd)^{2}}c_{N,d}(\mathcal{Z}^{(\mathrm{Gin})}_{N})^{d}\prod_{\ell=0}^{d-1}R_{\gamma_\ell}(t\sqrt{d}), \label{normstwice}
\end{equation}
where $R_{\gamma}(z)$ is defined in \eqref{Robject}, see also \eqref{ZN_Toeplitz}. It is also interesting to consider a second proof going explicitly via the determinant formula in \eqref{norms}. We have
\begin{equation}
\mathcal{Z}^{(\mathrm{Lem}_{d})}_{Nd}(t) = (Nd)!\det\bigg\{\int_{\mathbb{C}}\lambda^{j}\overline{\lambda}^{k}e^{-Nd\,V^{(d)}(\lambda,t)}d^{2}\lambda\bigg\}_{j,k=0}^{Nd-1}. \label{compmommat}
\end{equation}
We claim that most of the entries of the above moment matrix are zero due to symmetry. Indeed, because $V^{(d)}(\lambda,t)$ is invariant under the rotation $\lambda \to \lambda e^{2\pi i/d}$ we see that in a given $d \times d$ sub-block parameterized by $j = j'd+\ell$, $k=k'd+\ell'$, where $\ell = 0,\ldots,d-1$ and $\ell' = 0,\ldots,d-1$, one has
\begin{equation}
\int_{\mathbb{C}}\lambda^{j}\overline{\lambda}^{k}e^{-Nd\,V^{(d)}(\lambda,t)}d^{2}\lambda = e^{2i\pi(j'd-k'd)/d}e^{2i\pi(\ell-\ell')/d}\int_{\mathbb{C}}\lambda^{j}\overline{\lambda}^{k}e^{-Nd\,V^{(d)}(\lambda,t)}d^{2}\lambda.
\end{equation}
By the constraints on $\ell$ and $\ell'$, it is impossible for $(\ell-\ell')/d$ to be an integer unless $\ell=\ell'$. This shows that our moment matrix in \eqref{compmommat} consists of $d \times d$ diagonal blocks. Since all these blocks commute one can treat them as scalars and compute the determinant as for an ordinary matrix of scalars, followed by the determinant of the resulting diagonal matrix:
\begin{equation}
\begin{split}
\mathcal{Z}^{(\mathrm{Lem}_{d})}_{Nd}(t) &= (Nd)!\prod_{\ell=0}^{d-1}\det\bigg\{\int_{\mathbb{C}}|\lambda|^{2\ell}(\lambda^{d})^{j}(\overline{\lambda}^{d})^{k}e^{-Nd\,V^{(d)}(\lambda,t)}\,d^{2}\lambda\bigg\}_{j,k=0}^{N-1}\\
&= (Nd)!\prod_{\ell=0}^{d-1}\det\bigg\{d^{-1}\int_{\mathbb{C}}|u|^{\gamma_{\ell}}u^{j}\overline{u}^{k}e^{-Nd|u-t|^{2}+Nd|t|^{2}}\,d^{2}u\bigg\}_{j,k=0}^{N-1}\\
&= \frac{(Nd)!d^{-Nd}}{(N!)^{d}}\prod_{\ell=0}^{d-1}\int_{\mathbb{C}^{N}}\prod_{j=1}^{N}|u_{j}|^{\gamma_{\ell}}e^{-Nd|u_{j}-t|^{2}+Nd|t|^{2}}|\Delta(\vec{u})|^{2}\,d^{2}\vec{u},
\end{split}
\end{equation}
where we made the same change of variable $u=\lambda^{d}$ as before. Now shifting $u_{j} \to u_{j}+t$ followed by $u_{j} \to u_{j}/\sqrt{d}$ completes the proof.

\end{proof}


Combining this result with Lemma \ref{lem:redgin} gives the following:
\begin{corollary}
The partition function of the lemniscate ensemble given by \eqref{lempart} has an exact representation in terms of solutions of the $\sigma$-form of Painlev\'e V in \eqref{pvintro}. Specifically, we have
\begin{equation}
\begin{split}
\mathcal{Z}^{(\mathrm{Lem}_{d})}_{Nd}(t) &= e^{(Ntd)^{2}}c_{N,d}(\mathcal{Z}^{(\mathrm{Gin})}_{N})^{d}\prod_{\ell=0}^{d-1}R_{\gamma_\ell}(t\sqrt{d})\\
&= e^{(Ntd)^{2}-Nt^{2}\frac{d(d-1)}{2}}c_{N,d}(\mathcal{Z}^{(\mathrm{Gin})}_{N})^{d}\prod_{\ell=0}^{d-1}\left(R_{\gamma_l}(0)\,\mathrm{exp}\left(\int_{0}^{t^{2}dN}\frac{\sigma_{\gamma_\ell/2,N}(s)}{s}\,ds\right)\right),
\end{split}
\end{equation}
where the constants are given in Lemma \ref{lem:ginred2} and \eqref{angam}.
\end{corollary}
In the following, recall that $t_{\mathrm{c}} = \frac{1}{\sqrt{d}}$. Combining Lemma \ref{lem:redgin} with Theorem \ref{th:ww} gives an asymptotic expansion of the lemniscate partition function in the sub-critical phase.
\begin{corollary}
In the sub-critical regime $0 < t < t_{\mathrm{c}}$ with $t$ fixed, we have
\begin{equation}
\frac{\mathcal{Z}^{(\mathrm{Lem}_{d})}_{Nd}(t)}{\mathcal{Z}^{(\mathrm{Lem}_{d})}_{Nd}(0)} = e^{(Ntd)^{2}-Nt^{2}\frac{d(d-1)}{2}}\,(1+o(1)), \qquad N \to \infty. \label{lemsubcrit}
\end{equation}
\end{corollary}
In the critical phase, based on the previous computations and Theorem \ref{th:onecharge} we can make a conjecture for the expansion at $t=t_{\mathrm{c}}$.
\begin{conjecture}
\label{con:crit}
In the critical regime $t = t_{\mathrm{c}}-\frac{\tau t_{\mathrm{c}}}{\sqrt{N}}$ with $\tau \in \mathbb{R}$ fixed, we have
\begin{equation}
\frac{\mathcal{Z}^{(\mathrm{Lem}_{d})}_{Nd}(t)}{\mathcal{Z}^{(\mathrm{Lem}_{d})}_{Nd}(0)} = e^{(Ntd)^{2}-Nt^{2}\frac{d(d-1)}{2}}\,\prod_{\ell=0}^{d-1}F_{\gamma_{\ell}/2}(2\tau)\,(1+o(1)), \qquad N \to \infty, \label{lemcrit}
\end{equation}
where $F_{\gamma_\ell/2}(x)$ denotes the right-hand side of \eqref{Fk} with $k$ replaced with $\gamma_{\ell}/2$ in \eqref{p4}.
\end{conjecture}
It seems likely that the proof of Conjecture \ref{con:crit} above will require Riemann-Hilbert techniques and is postponed to a future investigation. In the super-critical phase $t > t_{\mathrm{c}}$, by Lemma \ref{lem:redgin} one requires asymptotics in the exterior region outside the circular law spectrum of the Ginibre ensemble, $|z|>1$. This corresponds to Laplace asymptotics of a random matrix linear statistic with a test function which is smooth on the support of the equilibrium measure. Indeed, as we explain in Appendix \ref{app:gff}, combining Lemma \ref{lem:redgin} and the results of \cite{RV07, AHM11} suggests that for $t > t_{\mathrm{c}}$ fixed, we have
\begin{equation}
\frac{\mathcal{Z}^{(\mathrm{Lem}_{d})}_{Nd}(t)}{(\mathcal{Z}^{(\mathrm{Lem}_{1})}_{N}(t\sqrt{d}))^{d}} = c_{N,d}\,\left(\frac{t}{t_{\mathrm{c}}}\right)^{N(d-1)}\left(1-\left(\frac{t_{\mathrm{c}}}{t}\right)^{2}\right)^{-\kappa_{d}}\,(1+o(1)), \qquad N \to \infty, \label{supercrit}
\end{equation}
where
\begin{equation}
\kappa_{d} :=\frac{1}{4}\sum_{\ell=0}^{d-1}\gamma_{\ell}^{2} = \frac{d(d-1)(2d-1)}{6d^{2}}.
\end{equation}
\begin{remark}
\label{rem:serfaty}
Asymptotic expansion of the partition function of planar models has received some attention recently. In the works \cite{BBNY16} and \cite{LS17} a partition function which includes \eqref{generalpf} as a special case has been studied (the Coulomb and Riesz gases), but their asymptotic expansion only gives terms up to order $N$, while the expansions here (albeit for specific models) include the constant term. In the physics literature, see \cite{ZW06} for interesting discussion and conjectures regarding the general form of this constant term.
\end{remark}
\begin{remark}
Note that the normalizations $\mathcal{Z}^{(\mathrm{Lem}_{d})}_{Nd}(0)$ and $(\mathcal{Z}^{(\mathrm{Lem}_{1})}_{N}(t\sqrt{d}))^{d}$ appearing in the denominators on the left-hand sides of \eqref{lemsubcrit}, \eqref{lemcrit} and \eqref{supercrit} involve radial weights and are simple to calculate explicitly.
\end{remark}

\subsection{Multiple charges near the boundary}
\label{se:boundarymerge}
\begin{proof}[Proof of Theorem \ref{th:multiplebdry}]
Proceeding as in the previous sections, we start with \eqref{akevern} under the microscopic scaling
\begin{equation}
x_{i} = z - \frac{u_{i}}{\overline{z}\sqrt{N}}, \qquad \overline{y_{i}} = \overline{z} - \frac{\overline{v_{i}}}{z\sqrt{N}} \label{edgescalecomplex}
\end{equation}
for $i=1,\ldots,k$, where $z$ is a point on the boundary, $|z|=1$. The polynomial part of kernel is
\begin{equation}
B_{N+k}(x,\overline{y}) = \frac{N}{\pi}\sum_{j=0}^{N+k-1}\frac{(x\overline{y}N)^{j}}{j!}.
\end{equation}
By Lemma 9.4 of \cite{BS09}, under the scaling \eqref{edgescalecomplex} we have, for any fixed $z \in S^{1}$, the following asymptotic formula uniformly in compact subsets of $u$ and $v$,
\begin{equation}
B_{N+k}(x,\overline{y}) \sim \frac{N}{2\pi}\,e^{N-\sqrt{N}(u+\overline{v})+u\overline{v}}\,\mathrm{erfc}\left(-\frac{u+\overline{v}}{\sqrt{2}}\right), \qquad N \to \infty,
\end{equation}
where $\mathrm{erfc}(z) := \frac{2}{\sqrt{\pi}}\int_{z}^{\infty}e^{-t^{2}}\,dt$ is the complementary error function. Regarding the other quantities in \eqref{akevern}, Stirling's formula gives $\prod_{j=0}^{k-1}h_{j+N} \sim \pi^{k}e^{-Nk}N^{-k/2}(2\pi)^{k/2}$ while the Vandermonde determinants rescale as $\Delta(\vec{x})\Delta(\vec{y}) = N^{-\frac{k(k-1)}{2}}\Delta(\vec{u})\Delta(\vec{v})$. Defining the \textit{error function kernel}
\begin{equation}
K_{\mathrm{erf}}(u,\overline{v}) := e^{-\frac{(u-\overline{v})^{2}}{2}}\,\mathrm{erfc}\left(-\frac{u+\overline{v}}{\sqrt{2}}\right)
\end{equation}
and inserting these results into \eqref{akevern}, using \eqref{mergeakevern}, we get the formula
\begin{equation}
\begin{split}
\mathbb{E}&\left(\prod_{i=1}^{k}\det(G_{N}-x_i)\det(G_{N}^{\dagger}-\overline{y_i})\right) \sim \mathrm{exp}\left(-\sqrt{N}\sum_{i=1}^{k}(u_i+\overline{v_i})+\frac{k^{2}}{2}\log(N)+\sum_{i=1}^{k}\frac{u_{i}^{2}+\overline{v_{i}}^{2}}{2}\right)\\
&\times\left(\frac{\pi}{2}\right)^{k/2}\,\frac{\det\bigg\{K_{\mathrm{erf}}(u_i,\overline{v_j})\bigg\}_{i,j=1}^{k}}{\Delta(\vec{u})\Delta(\vec{v})}, \qquad N \to \infty.
\end{split}
\end{equation}
The relation to non-intersecting paths goes via the integral representation
\begin{equation}
K_{\mathrm{erf}}(u,\overline{v}) =2\sqrt{2\pi}\,\int_{0}^{\infty}p_{1/2}(u,s)p_{1/2}(s,\overline{v})\,ds
\end{equation}
in terms of the Brownian motion transition density
\begin{equation}
p_{t}(u,v) = \frac{1}{\sqrt{2\pi t}}\,e^{-\frac{(u-v)^{2}}{2t}}.
\end{equation}
Consequently, by Lemma \ref{lem:andre} we can immediately identify
\begin{equation}
\det\bigg\{K_{\mathrm{erf}}(u_i,v_j)\bigg\}_{i,j=1}^{k} = \frac{2^{k}(2\pi)^{k/2}}{k!}\,\mathcal{Z}_{1/2}(\vec{u},\vec{v},\mathbb{R}_{+}),
\end{equation}
and \eqref{multiplecharge} follows.
\end{proof}

\section{Bulk asymptotics: proofs}
\label{se:bulk}
The purpose of this section will be to prove our results on bulk asymptotics, namely Theorems \ref{th:twocharge} and \ref{th:hcizbulk} which deal with integer charge exponents and Theorem \ref{th:nonintegerbulk} which deals with a particular case of non-integer exponents. The proofs are both based on formula \eqref{akevern} but proceed in different ways, one exploits the Harish-Chandra Itzykson Zuber integral while the other exploits the exactly solvable structure of the \textit{induced Ginibre ensemble} studied in \cite{FBKSZ}. Our proof of Theorem \ref{th:nonintegerbulk} can be interpreted as the computation of the $1$-point correlator for the induced Ginibre ensemble.
\begin{proof}[Proof of Theorems \ref{th:twocharge} and \ref{th:hcizbulk}]
We focus to begin with on the proof of Theorem \ref{th:twocharge}. Taking \eqref{akevern} as our starting point, with the Ginibre weight $w(\lambda) = e^{-N|\lambda|^{2}}$ the usual formulas apply and we have $p_{j}(\lambda) = \lambda^{j}$, $h_{j} = N^{-j-1}\pi j!$, and the asymptotics
\begin{equation}
\prod_{j=0}^{k-1}h_{j+N} \sim \pi^{k}\,(2\pi)^{k/2}\,e^{-Nk}\,N^{-k/2}, \qquad N \to \infty.
\end{equation}
The polynomial part of the correlation kernel is
\begin{equation}
B_{N+k}(x,\overline{y}) = \frac{N}{\pi}\sum_{j=0}^{N+k-1}\frac{(x\overline{y}N)^{j}}{j!}.
\end{equation}
The idea of the proof will be to merge the variables in \eqref{akevern} in two distinct groups, the first group creating a singularity with strength $k_{1}$ and the second with strength $k_{2}$ and $k=k_{1}+k_{2}$. It is convenient to first impose the microscopic scaling 
\begin{equation}
x_{j} = z+\frac{u_{j}}{\sqrt{N}}, \qquad \overline{y_{j}} = \overline{z} + \frac{\overline{v_{j}}}{\sqrt{N}}, \qquad j=1,\ldots,k
\end{equation}
and then perform the merging (still keeping $N$ finite)
\begin{equation}
(u_1,\ldots,u_{k_1}) \to (u_1,\ldots,u_1), \qquad (u_{k_{1}+1},\ldots,u_{k_1+k_2}) \to (u_2,\ldots,u_2)
\end{equation}
and
\begin{equation}
(\overline{v_1},\ldots,\overline{v_{k_1}}) \to (\overline{v_1},\ldots,\overline{v_1}), \quad (\overline{v_{k_{1}+1}},\ldots,\overline{v_{k_1+k_2}}) \to (\overline{v_{2}},\ldots,\overline{v_{2}}).
\end{equation}
Doing this, we can straightforwardly write down the exact identity
\begin{equation}
\begin{split}
&\mathbb{E}\left(|\det(G_{N}-z_1)|^{2k_1}|\det(G_{N}-z_2)|^{2k_2}\right) = e^{Nk|z|^{2}+\sqrt{N}k_{1}(u_{1}\overline{z}+\overline{u_1}z)+\sqrt{N}k_{2}(u_{2}\overline{z}+\overline{u_2}z)}\\
&\times \left(\prod_{j=0}^{k-1}h_{j+N}\right)\,N^{k}\pi^{-k}\,N^{\frac{k(k-1)}{2}}\,\\
&\times \lim_{v_{2}\to u_{2},v_{1} \to u_{1}}\lim_{\substack{(u_1,\ldots,u_{k_1}) \to (u_1,\ldots,u_1),(u_{k_1+1},\ldots,u_{k_1+k_2}) \to (u_2,\ldots,u_2)\\
(\overline{v_{1}},\ldots,\overline{v_{k_1}}) \to (\overline{v_1},\ldots,\overline{v_{1}}), (\overline{v_{k_1+1}},\ldots,\overline{v_{k_1+k_2}}) \to (\overline{v_{2}},\ldots,\overline{v_{2}})}}\frac{\det\bigg\{f_{N+k}(u_i,\overline{v_j})\bigg\}_{i,j=1}^{k}}{\Delta(\vec{u})\overline{\Delta(\vec{v})}}
\end{split} \label{limitfnkbulk}
\end{equation}
where we used that $\Delta(\vec{x})\overline{\Delta(\vec{y})} = N^{-\frac{k(k-1)}{2}}\Delta(\vec{u})\overline{\Delta(\vec{v})}$ and defined
\begin{equation}
f_{N+k}(u,\overline{v}) = e^{-N|z|^{2}-\sqrt{N}(u\overline{z}+\overline{v}z)}\sum_{j=0}^{N+k-1}\frac{\left(N(z+\frac{u}{\sqrt{N}})(\overline{z}+\frac{\overline{v}}{\sqrt{N}})\right)^{j}}{j!}, \qquad f(u,\overline{v}) = e^{u\overline{v}}
\end{equation}
in such a way that $\lim_{N \to \infty}f_{N+k}(u,\overline{v}) = f(u,\overline{v})$ uniformly in compact subsets of $u$ and $v$ (see \textit{e.g.} \cite[Lemma 9.2]{BS09} or \cite[Proposition 2]{BW07}). The degenerate limit in \eqref{limitfnkbulk} can be written as the determinant of a block matrix involving partial derivatives of $f_{N+k}(u,\overline{v})$ with respect to the $u$ and $v$ variables. Indeed, all derivatives of $f_{N}(u,\overline{v})$ to any finite order converge uniformly to the corresponding derivatives of $f(u,\overline{v})$ as $N \to \infty$, and thus the various limiting operations can be interchanged. Therefore, the limit $N \to \infty$ of the ratio of determinants in \eqref{limitfnkbulk} is equal to
\begin{equation}
\lim_{\substack{(u_1,\ldots,u_{k_1}) \to (u_1,\ldots,u_1),(u_{k_1+1},\ldots,u_{k_1+k_2}) \to (u_2,\ldots,u_2)\\
(\overline{v_{1}},\ldots,\overline{v_{k_1}}) \to (\overline{v_1},\ldots,\overline{v_{1}}), (\overline{v_{k_1+1}},\ldots,\overline{v_{k_1+k_2}}) \to (\overline{v_{2}},\ldots,\overline{v_{2}})}}\frac{\det\bigg\{e^{u_i\,\overline{v_j}}\bigg\}_{i,j=1}^{k}}{\Delta(\vec{u})\overline{\Delta(\vec{v})}}.
\end{equation}
More generally, the same considerations show that
\begin{equation}
\begin{split}
\mathbb{E}\left(\prod_{i=1}^{k}\det(G_{N}-x_{i})\det(G_{N}^{\dagger}-\overline{y_{i}})\right) &= \mathrm{exp}\left(Nk(|z|^{2}-1)+\sqrt{N}\sum_{i=1}^{k}(z\overline{v_{i}}+\overline{z}u_{i})+\frac{k^{2}}{2}\log(N)\right)\\
& \times (2\pi)^{k/2}\frac{\det\bigg\{e^{u_{i}\overline{v_{j}}}\bigg\}_{i,j=1}^{k}}{\Delta(\vec{u})\Delta(\vec{\overline{v}})}\left(1+o(1)\right) \label{HCIZform}
\end{split}
\end{equation}
where the asymptotics on the right-hand side are defined by taking a limit if some of the points in the vectors $\vec{u}$ or $\vec{v}$ coincide. To conclude the proof of Theorem \ref{th:hcizbulk} we recognise the right-hand side of \eqref{HCIZform} in terms of the Harish-Chandra Itzykson Zuber integral \cite{IZ80}, which is the exact formula
\begin{equation}
\frac{\det\bigg\{e^{u_{i}\overline{v}_{j}}\bigg\}_{i,j=1}^{k}}{\Delta(\vec{u})\Delta(\vec{\overline{v}})} = \frac{1}{G(1+k)}\int_{U(k)}\,e^{\mathrm{Tr}(UAU^{\dagger}\overline{B})}\,d\mu(U), \label{HCIZformula}
\end{equation}
where $d\mu(U)$ is the normalized Haar measure on the group $U(k)$ of $k \times k$ unitary matrices, and $A = \mathrm{diag}(u_1,\ldots,u_k)$, $\overline{B} = \mathrm{diag}(\overline{v_1},\ldots,\overline{v_k})$ are diagonal matrices. The rather intricate merging in \eqref{limitfnkbulk} can now be performed directly by creating degeneracies in the diagonal matrices $A$ and $\overline{B}$. For Theorem \ref{th:twocharge} we simply put
\begin{equation}
\begin{split}
A &= \mathrm{diag}(\overbrace{u_1,\ldots,u_1}^{k_1},\overbrace{u_2,\ldots,u_2}^{k_2})\\
\overline{B} &=  \mathrm{diag}(\overbrace{\overline{v_1},\ldots,\overline{v_1}}^{k_1},\overbrace{\overline{v_2},\ldots,\overline{v_2}}^{k_2})
\end{split}
\end{equation}
where eventually we will set $u_1=v_1$ and $u_2=v_2$. This motivates us to split the unitary matrix $U$ into four sub-blocks. Let us write
\begin{equation}
U = \begin{pmatrix} a_{k_1 \times k_1} & b_{k_1 \times k_2}\\ c_{k_2 \times k_1} & d_{k_2 \times k_2} \end{pmatrix}.
\end{equation}
Then a direct calculation using the unitarity of $U$ shows that
\begin{equation}
\mathrm{Tr}(UAU^{\dagger}\overline{B}) = u_{1}\overline{v_{1}}k_{1}+u_{2}\overline{v_{2}}k_{2}-[(u_{2}-u_{1})(\overline{v_{2}}-\overline{v_{1}})]\mathrm{Tr}(cc^{\dagger}), \label{explicit-trace}
\end{equation}
where we abbreviated the matrix $c := c_{k_2 \times k_1}$. To proceed, we need to know how to project the Haar measure in \eqref{HCIZformula} onto the sub-block $c$. This turns out to be a well studied problem (not least because of various physical applications, see \cite{F06}). For $k_{1} \geq k_{2}$, the eigenvalues $t_{1},\ldots,t_{k_2}$ of $cc^{\dagger}$ lie in the unit interval $[0,1]$ and have the joint probability density function \cite{F06,Col05}
\begin{equation}
P(t_1,\ldots,t_{k_2}) = \frac{1}{C^{(\mathrm{JUE}_{k_1-k_2,0})}_{k_2}}\prod_{j=1}^{k_2}t_{j}^{k_{1}-k_{2}}\,\Delta^{2}(\vec{t}) \label{jacobidistn}
\end{equation}
where $C^{(\mathrm{JUE}_{k_1-k_2,0})}_{k_2}$ is a normalization constant for the Jacobi Unitary Ensemble (JUE),
\begin{equation}
C^{(\mathrm{JUE}_{k_1-k_2,0})}_{k_2} := \int_{[0,1]^{k_2}}\prod_{j=1}^{k_2}t_{j}^{k_1-k_2}\Delta^{2}(\vec{t})\,d\vec{t}.
\end{equation}
Inserting \eqref{explicit-trace} into \eqref{HCIZformula} and using \eqref{jacobidistn}, we find
\begin{equation}
\begin{split}
\int_{U(k)}e^{\mathrm{Tr}(UAU^{\dagger}\overline{B})}\,d\mu(U) &=  e^{u_1\overline{v_1}k_{1}+u_{2}\overline{v_{2}}k_{2}}\,\mathbb{E}(e^{-s\mathrm{Tr}(cc^{\dagger})}).\\
&= \frac{ e^{u_1\overline{v_1}k_{1}+u_{2}\overline{v_{2}}k_{2}}}{C^{(\mathrm{JUE}_{k_1-k_2,0})}_{k_2}}\,\int_{[0,1]^{k_2}}\prod_{j=1}^{k_2}t_{j}^{k_1-k_2}e^{-st_{j}}\Delta^{2}(\vec{t})d\vec{t},
\end{split}
\end{equation}
where we defined the variable $s=(u_2-u_1)(\overline{v_2}-\overline{v_1})$. Finally, we express this last integral as a gap probability in the LUE. For this we specialize to $u_1 = \overline{v_1}$ and $u_2=\overline{v_2}$ so that $s = |u_{2}-u_{1}|^{2} = N|z_2-z_1|^{2}$. We have
\begin{equation}
\begin{split}
&\frac{1}{C^{(\mathrm{JUE}_{k_1-k_2,0})}_{k_2}}\int_{[0,1]^{k_2}}\prod_{j=1}^{k_2}t_{j}^{k_1-k_2}e^{-st_{j}}\,\Delta^{2}(\vec{t})\,d\vec{t}\\
&=\frac{C_{k_2}^{(\mathrm{LUE}_{k_1-k_2})}}{C^{(\mathrm{JUE}_{k_1-k_2,0})}_{k_2}}s^{-k_1k_2}\mathbb{P}(\lambda_{\mathrm{max}}^{(\mathrm{LUE}_{k_2,k_1})} < s).
\end{split}
\end{equation}
The normalization constants here are explicitly known, see equations \eqref{LUEnormconst}, \eqref{JUEnormconst} and \cite[Chapter 4]{Forrester_loggas}. We have
\begin{equation}
\begin{split}
\frac{C_{k_2}^{(\mathrm{LUE}_{k_1-k_2})}}{C^{(\mathrm{JUE}_{k_1-k_2,0})}_{k_2}} &= \prod_{j=0}^{k_2-1}\frac{\Gamma(1+k_1+j)}{\Gamma(1+j)}\\
&= \frac{G(k_1+k_2+1)}{G(k_1+1)G(k_2+1)}, \label{rationorms}
\end{split}
\end{equation}
the numerator of which exactly cancels the denominator in \eqref{HCIZformula}. Therefore, we get the following asymptotic formula, uniformly in compact subsets of $u_1$ and $u_2$,
\begin{equation}
\begin{split}
\mathbb{E}&\left(|\det(G_{N}-z_1)|^{2k_1}|\det(G_{N}-z_2)|^{2k_2}\right) = e^{k_{1}N(|z_1|^{2}-1)+k_{2}N(|z_2|^{2}-1)+\frac{k_1^{2}}{2}\log(N)+\frac{k_{2}^{2}}{2}\log(N)}\\
&\times e^{-2k_1k_2\log|z_2-z_1|}\,\frac{(2\pi)^{(k_1+k_2)/2}}{G(1+k_1)G(1+k_2)}F_{k_1,k_2}(N|z_2-z_1|^{2})(1+o(1)), \qquad N \to \infty,
\end{split}
\end{equation}
where $F_{k_1,k_2}(s)$ is the distribution function of the largest eigenvalue in a $k_2 \times k_2$ Laguerre ensemble with parameter $k_1-k_2$ (equivalently, $k_1$ degrees of freedom):
\begin{equation}
F_{k_1,k_2}(s) = \mathbb{P}\left(\lambda_{\mathrm{max}}^{(\mathrm{LUE}_{k_2,k_1})} < s\right).
\end{equation}
This concludes the proof of Theorem \ref{th:twocharge}.
\end{proof}

\begin{proof}[Proof of Theorem \ref{th:nonintegerbulk}]
With the assumptions of Theorem \ref{th:nonintegerbulk}, the $2$-point correlator on the left-hand side of \eqref{corr2} simplifies to
\begin{equation}
\begin{split}
&\mathbb{E}\left(|\det(G_{N})|^{2\gamma_{1}}\,|\det(G_{N}-z_2)|^{2k_2}\right)\\
&= \frac{1}{\mathcal{Z}^{(\mathrm{Gin})}_{N}}\int_{\mathbb{C}^{N}}\prod_{j=1}^{N}|\lambda_{j}|^{2\gamma_{1}}\,|\lambda_{j}-z_2|^{2k_2}\,e^{-N|\lambda_j|^{2}}\,|\Delta(\vec{\lambda})|^{2}\,d^{2}\vec{\lambda} \label{corr2center}
\end{split}
\end{equation}
where $z_{2} = \frac{u_{2}}{\sqrt{N}}$ and $u_2$ is fixed in the microscopic scaling. By the radial symmetry, we can assume without loss of generality that $u_2>0$ throughout. The point of this proof, in contrast to the proof of Theorem \ref{th:twocharge}, is that $\gamma_{1}$ is \textit{not constrained to be a positive integer}. We assume that $\gamma_{1} \geq k_2>0$ and $k_2 \in \mathbb{N}$.

We now proceed by viewing \eqref{corr2center} as a $1$-point correlator for the reference weight $w(\lambda) = |\lambda|^{2\gamma_{1}}e^{-N|\lambda|^{2}}$ and apply the techniques of the previous sections. Adapting the proof of Proposition \ref{prop:lue}, we again have a rotationally invariant weight, implying $p_{j}(\lambda) = \lambda^{j}$, but now the norms in \eqref{norms} are
\begin{equation}
h_{j} = \int_{\mathbb{C}}|\lambda|^{2j+2\gamma_{1}}e^{-N|\lambda|^{2}}\,d^{2}\lambda = \pi N^{-j-\gamma_{1}-1}\Gamma(j+\gamma_{1}+1). \label{indnorms}
\end{equation}
Under the microscopic scaling, the points in \eqref{akevern} take the form
\begin{equation}
x_{i} = \frac{u_{i}}{\sqrt{N}}, \qquad y_{i} = \frac{v_{i}}{\sqrt{N}}, \qquad i=1,\ldots,k \label{microinduced}
\end{equation}
and the product of Vandermonde determinants in the denominator of \eqref{akevern} rescale to $\Delta(\vec{x})\Delta(\vec{y}) = N^{-\frac{k_{2}(k_{2}-1)}{2}}\Delta(\vec{u})\Delta(\vec{v})$. From the explicit formulae for $p_{j}(\lambda)$ and $h_{j}$, we get the polynomial part of the kernel
\begin{equation}
B_{N+k}(x,y) = \frac{N^{\gamma_{1}+1}}{\pi}\sum_{j=0}^{N+k-1}\frac{(xyN)^{j}}{\Gamma(j+\gamma_{1}+1)}. \label{inducedkernel}
\end{equation}
Defining the functions
\begin{equation}
f_{N}(u,v) := v^{\gamma_{1}}\sum_{j=0}^{N-1}\frac{(uv)^{j}}{\Gamma(j+\gamma_{1}+1)}, \qquad f(u,v):=v^{\gamma_{1}}\sum_{j=0}^{\infty}\frac{(uv)^{j}}{\Gamma(j+\gamma_{1}+1)}
\end{equation}
we first observe that \eqref{akevern} and \eqref{mergeakevern} give the \textit{exact identity}
\begin{equation}
\begin{split}
&\mathbb{E}\left(|\det(G_{N})|^{2\gamma_{1}}|\det(G_{N}-z_2)|^{2k_{2}}\right) = \frac{\mathcal{Z}^{(\mathrm{Ind}_{\gamma_{1}})}_{N}}{\mathcal{Z}^{(\mathrm{Gin})}_{N}}\,N^{k_{2}\gamma_{1}}\pi^{-k_{2}}u_{2}^{-\gamma_{1} k_{2}}N^{k_{2}(k_{2}+1)/2}\\
&\times\prod_{j=0}^{k_{2}-1}h_{j+N}\, \frac{1}{\prod_{j=0}^{k_{2}-1}(j!)^{2}}\,\lim_{u \to u_2, v \to u_2}\det\bigg\{\frac{\partial^{i+j-2}f_{N+k_{2}}(u,v)}{\partial u^{i-1}\partial v^{j-1}}\bigg\}_{i,j=1}^{k_{2}}, \label{inducedpartials}
\end{split}
\end{equation}
where the partition function of the induced Ginibre ensemble is
\begin{equation}
\mathcal{Z}^{(\mathrm{Ind}_{\gamma_{1}})}_{N} := \int_{\mathbb{C}^{N}}\prod_{j=1}^{N}|\lambda_{j}|^{2\gamma_{1}}e^{-N|\lambda_{j}|^{2}}\,|\Delta(\vec{\lambda})|^{2}\,d^{2}\vec{\lambda} = N!\prod_{j=0}^{N-1}(\pi N^{-j-\gamma_{1}-1}\Gamma(j+\gamma_{1}+1)). \label{indgpart}
\end{equation}
Now the large-$N$ asymptotics of \eqref{inducedpartials} can be derived, exploiting the fact that the convergence $f_{N}(u,v) \to f(u,v)$ as $N \to \infty$ is uniform. Furthermore, all mixed partial derivatives of $f_{N}(u,v)$ converge uniformly to those of $f(u,v)$ and this allows the interchange of the various limiting operations. Consequently, it suffices to compute the determinant in \eqref{inducedpartials} with $f_{N}(u,v)$  replaced with the limiting kernel $f(u,v)$, for which we have the integral representation
\begin{equation}
f(u,v) = \frac{1}{\Gamma(\gamma_{1})}\,\int_{0}^{v}e^{ut}(v-t)^{\gamma_{1}-1}\,dt.
\end{equation}
The needed partial derivatives of the kernel can be computed using the Leibniz formula for differentiating under the integral, using that the integrand vanishes at the upper end point. This gives
\begin{equation}
\begin{split}
\frac{\partial^{i+j-2}f(u,v)}{\partial u^{i-1}\partial v^{j-1}} &= \frac{1}{\Gamma(\gamma_{1}-(j-1))}\int_{0}^{v}t^{i-1}e^{ut}(v-t)^{\gamma_{1}-1-(j-1)}\,dt\\
&= \frac{v^{\gamma_{1}+i-j}(-1)^{i-1}e^{uv}}{\Gamma(\gamma_{1}-(j-1))}\int_{0}^{1}t^{\gamma_{1}-1}e^{-uvt}(t-1)^{i-1}t^{-(j-1)}\,dt.
\end{split}
\end{equation}
By Lemma \ref{lem:andre},
\begin{equation}
\begin{split}
&\det\bigg\{\frac{\partial^{i+j-2}f(u,v)}{\partial u^{i-1}\partial v^{j-1}} \bigg\}_{i,j=1}^{k_{2}}\\
&= \frac{v^{\gamma_{1} k_{2}}e^{uvk_{2}}}{k_{2}!}\prod_{j=0}^{k_{2}-1}\frac{(-1)^{j}}{\Gamma(\gamma_{1}-j)}\int_{[0,1]^{k_{2}}}\prod_{j=1}^{k_{2}}t_{j}^{\gamma_{1}-1}e^{-uvt_{j}}\det\bigg\{(t_{i}-1)^{j-1}\bigg\}_{i,j=1}^{k_2}\det\bigg\{t_{i}^{-(j-1)}\bigg\}_{i,j=1}^{k_2}\,d\vec{t}\\
&=\frac{v^{\gamma_{1} k_{2}}\,(uv)^{-\gamma_{1} k_{2}}\,e^{uvk_{2}}}{k_{2}!}\prod_{j=0}^{k_{2}-1}\frac{1}{\Gamma(\gamma_{1}-j)}\int_{[0,uv]^{k_2}}\prod_{j=1}^{k_2}t_{j}^{\gamma_{1}-k_{2}}e^{-t_{j}}\,\Delta^{2}(\vec{t})\,d\vec{t}, \label{chargecollideint}
\end{split}
\end{equation}
where we now spot the joint distribution of eigenvalues of the Laguerre ensemble appearing in \eqref{chargecollideint}. We have therefore arrived at the determinant identity
\begin{equation}
\det\bigg\{\frac{\partial^{i+j-2}f(u,v)}{\partial u^{i-1}\partial v^{j-1}} \bigg\}_{i,j=1}^{k_{2}} = \frac{v^{\gamma_{1} k_{2}}\,(uv)^{-\gamma_{1} k_{2}}\,G(1+k_{2})\,e^{uvk_{2}}}{k_{2}!}\,\mathbb{P}(\lambda^{(\mathrm{LUE}_{k_2,\gamma_1})}_{\mathrm{max}} < uv). \label{chargecollidelaguerre}
\end{equation}
It remains just to deal with the constants in \eqref{inducedpartials}. Writing the constant \eqref{indgpart} in terms of the Barnes G-function and applying their asymptotic expansion (see \cite[Eq. 5.17.5]{NIST:DLMF}) gives,
\begin{equation}
\frac{\mathcal{Z}^{(\mathrm{Ind}_{\gamma_{1}})}_{N}}{\mathcal{Z}^{(\mathrm{Gin})}_{N}} = \frac{1}{G(1+\gamma_{1})}\,e^{-\gamma_{1} N+\frac{\gamma_{1}^{2}}{2}\log(N)}(2\pi)^{\gamma_{1}/2}\,(1+o(1)), \qquad N \to \infty, \label{barnesasy}
\end{equation}
while Stirling's formula gives
\begin{equation}
\prod_{j=0}^{k_{2}-1}h_{j+N} \sim \pi^{k_{2}}\,(2\pi)^{k_{2}/2}\,e^{-Nk_{2}}N^{-k_{2}/2}, \qquad N \to \infty. \label{indstirling}
\end{equation}
Now inserting \eqref{indstirling}, \eqref{barnesasy} and \eqref{chargecollidelaguerre} into \eqref{inducedpartials} gives the formula
\begin{equation}
\begin{split}
&\mathbb{E}\left(|\det(G_{N})|^{2\gamma_{1}}|\det(G_{N}-z_2)|^{2k_{2}}\right) = e^{u_{2}^{2}k_{2}-(\gamma_{1}+k_{2})N+\frac{\gamma_{1}^{2}}{2}\log(N)+\frac{k_{2}^{2}}{2}\log(N)}\\
&\times e^{-2k_{2}\gamma_{1}\,\log(u_2/\sqrt{N})}\frac{(2\pi)^{\frac{\gamma_{1}+k_{2}}{2}}}{G(1+k_2)G(1+\gamma_{1})}\,\mathbb{P}(\lambda^{(\mathrm{LUE}_{k_2,\gamma_1})}_{\mathrm{max}} < u_{2}^{2})\,\left(1+o(1)\right), \qquad N \to \infty,
\label{corr2proved}
\end{split}
\end{equation}
which is precisely of the form \eqref{corr2}, with $k_1 = \gamma_{1}$ and $z=z_{1}=0$.
\end{proof}

\section{Truncated CUE: finite $N$ and boundary asymptotics}
\label{se:tcue}
The purpose of this section is to prove analogous results for the $1$-point function as obtained in Section \ref{se:ginibre}. Therefore, we consider the moments
\begin{equation}
R_{2k}(z) = \mathbb{E}\left(|\det(T-z)|^{2k}\right) \label{tcuemoments}
\end{equation}
where $T$ is an $N \times N$ truncation of a Haar distributed unitary matrix of size $M \times M$. We begin by showing that \eqref{tcuemoments} is expressible in terms of solutions of Painlev\'e VI via a duality with the Jacobi Unitary Ensemble. Then we will study the boundary asymptotics and show that they are expressed in terms of Painlev\'e V. This will prove Theorems \ref{th:tcuefiniten} and \ref{th:tcueonecharge}. As in Section \ref{se:ginibre}, we also investigate the continuation of Theorem \ref{th:tcuefiniten} off the integers, that is we will prove that setting $k=\frac{\gamma}{2}$ in \eqref{JUEprobab} and \eqref{pvirep} gives the correct interpretation of \eqref{tcuemoments}. Along the way, we will study planar orthogonal polynomials associated to the weight
\begin{equation}
w(\lambda) = |\lambda-z|^{\gamma}(1-|\lambda|^{2})^{M-N-1}, \qquad |\lambda| \leq 1.
\end{equation}
\subsection{Duality with the JUE}
We begin by deriving the analogue of identity \eqref{LUEduality}. Recall the construction of the Jacobi Unitary Ensemble (JUE): take two independent $k$-dimensional LUE matrices (Wishart matrices) $W_{1}$ with parameter $\alpha$ and $W_{2}$ with parameter $\beta$ and form the ratio $J = W_{1}(W_{1}+W_{2})^{-1}$. The eigenvalues $t_{1},\ldots,t_{k}$ of $J$ lie in the unit interval $[0,1]$ and have the joint probability density function
\begin{equation}
P(t_1,\ldots,t_k) = \frac{1}{C^{(\mathrm{JUE}_{\alpha,\beta})}_{k}}\prod_{j=1}^{k}t_{j}^{\alpha}(1-t_{j})^{\beta}\,\Delta^{2}(\vec{t}),
\end{equation}
where the normalization constant (see \textit{e.g.} \cite[Chapter 4]{Forrester_loggas}) is given explicitly by
\begin{equation}
C^{(\mathrm{JUE}_{\alpha,\beta})}_{k} := \int_{[0,1]^{k}}\prod_{j=1}^{k}t_{j}^{\alpha}(1-t_{j})^{\beta}\,\Delta^{2}(\vec{t})\,d\vec{t}= \prod_{j=0}^{k-1}\frac{\Gamma(\alpha+j+1)\Gamma(\beta+j+1)\Gamma(j+2)}{\Gamma(\alpha+\beta+k+j+1)}. \label{JUEnormconst}
\end{equation}
\begin{proposition}
\label{prop:JUEduality}
Consider the expectation with respect to the truncated unitary ensemble measure in \eqref{zspdf} with $\kappa = M-N$. Let $x$ and $y$ be any complex numbers and let $k \in \mathbb{N}$. We have
\begin{equation}
\begin{split}
&\mathbb{E}\left(\det(T-x)^{k}\det(T^{\dagger}-y)^{k}\right) \\
&= \frac{1}{C^{(\mathrm{JUE}_{\kappa+N,0})}_{k}}\int_{[0,1]^{k}}\prod_{j=1}^{k}t_{j}^{\kappa}(1+(xy-1)t_{j})^{N}\,\Delta^{2}(\vec{t})\,d\vec{t}. \label{jueduality}
\end{split}
\end{equation}
\end{proposition}

\begin{proof}[Proof of Proposition \ref{prop:JUEduality}, Theorem \ref{th:tcuefiniten} and Theorem \ref{th:tcueonecharge}]
Compared with the Ginibre case in Proposition \ref{prop:lue}, the main difference here is that we now work with the weight $w(\lambda) = (1-|\lambda|^{2})^{\kappa-1}$ supported on the disc $\lambda \in \mathbb{D}$. As before the corresponding orthogonal polynomials are $p_{j}(\lambda) = \lambda^{j}$ but now the normalizations are
\begin{equation}
h_{j} := \int_{\mathbb{D}}|\lambda|^{2j}\,(1-|\lambda|^{2})^{\kappa-1}\,d^{2}\lambda = \pi\frac{j!\Gamma(\kappa)}{\Gamma(j+\kappa+1)}. \label{hjtcue}
\end{equation}
The polynomial part of the kernel follows as
\begin{equation}
\begin{split}
&B_{N+k}(x,y) = \frac{1}{\pi\Gamma(\kappa)}\sum_{j=0}^{N+k-1}\frac{\Gamma(j+\kappa+1)(xy)^{j}}{\pi j!}\\
&= x^{N+k}\frac{\Gamma(N+k+\kappa+1)}{\pi\,\Gamma(\kappa)\Gamma(N+k)}\,\int_{0}^{\infty}(tx+1)^{-(N+k)-\kappa-1}(t+y)^{N+k-1}\,dt,\\ \label{intreptcue}
\end{split}
\end{equation}
where for now (and without loss of generality) we will assume that $x$ and $y$ are positive reals. The integral representation \eqref{intreptcue} follows \textit{e.g.} from Euler's integral representation for the hypergeometric function \cite[15.6.1]{NIST:DLMF}. Its advantage here is that we have separated the $x$ and $y$ variables in the integrand. Indeed, considering the function 
\begin{equation}
A^{(\kappa)}_{N+k}(x,y) := \frac{\Gamma(N+k+\kappa+1)}{\Gamma(N+k)}\int_{0}^{\infty}(tx+1)^{-N-k-\kappa-1}(t+y)^{N+k-1}\,dt,
\end{equation}
we can directly calculate the derivatives required on the right-hand side of \eqref{mergeakevern}:
\begin{equation}
\begin{split}
\frac{\partial^{i+j-2}A^{(\kappa)}_{N+k}(x,y)}{\partial x^{i-1} \partial y^{j-1}} &= \frac{\Gamma(N+k+\kappa+i)(-1)^{i-1}}{\Gamma(N+k-(j-1))}\\
&\times \int_{0}^{\infty}t^{i-1}(tx+1)^{-N-k-\kappa-1-(i-1)}(t+y)^{N+k-1-(j-1)}\,dt.
\end{split}
\end{equation}
By Lemma \ref{lem:andre} we get
\begin{equation}
\begin{split}
&\det\bigg\{\frac{\partial^{i+j-2}A^{(\kappa)}_{N+k}(x,y)}{\partial x^{i-1}\partial y^{j-1}}\bigg\}_{i,j=1}^{k} \\
&= e_{N,k,\kappa}\int_{[0,\infty)^{k}}\det\bigg\{t_{i}^{j-1}(t_{i}x+1)^{-N-k-j-\kappa}\bigg\}_{i,j=1}^{k}\,\det\bigg\{(t_{i}+y)^{N+k-j}\bigg\}_{i,j=1}^{k}\,d\vec{t}\\
&= (-1)^{\frac{k(k-1)}{2}}e_{N,k,\kappa}\int_{[0,\infty)^{k}}\,\prod_{j=1}^{k}(t_{j}x+1)^{-N-\kappa-2k}(t_{j}+y)^{N}\,\Delta^{2}(\vec{t})\,d\vec{t}, \label{tcuedets}\\
\end{split}
\end{equation}
where
\begin{equation}
e_{N,k,\kappa} := \frac{(-1)^{\frac{k(k-1)}{2}}}{k!}\,\prod_{j=1}^{k}\left(\frac{\Gamma(N+k+\kappa+j)}{\Gamma(N+j)}\right).
\end{equation}
To obtain the third line in \eqref{tcuedets} we made use of the determinant identities
\begin{equation}
\begin{split}
\det\bigg\{(t_{i}+y)^{k-j}\bigg\}_{i,j=1}^{k} &= (-1)^{\frac{k(k-1)}{2}}\Delta(\vec{t}),\\
\det\bigg\{(t_{i}^{j-1}(t_{i}x+1)^{-k-j}\bigg\}_{i,j=1}^{k} &= \prod_{j=1}^{k}(t_{j}x+1)^{-2k}\Delta(\vec{t}).
\end{split}
\end{equation}
After sending $t_{j} \to t_{j}/x$ for each $j=1,\ldots,k$ in \eqref{tcuedets} and some further standard manipulations, we arrive at the identity
\begin{equation}
\begin{split}
&\det\bigg\{\frac{\partial^{i+j-2}A^{(\kappa)}_{N+k}(x,y)}{\partial x^{i-1}\partial y^{j-1}}\bigg\}_{i,j=1}^{k} \\
&=(-1)^{\frac{k(k-1)}{2}} e_{N,k,\kappa}\,x^{-k(N+k)}\int_{[0,1]^{k}}\prod_{j=1}^{k}t_{j}^{\kappa}(1+(xy-1)t_{j})^{N}\,\Delta^{2}(\vec{t})\,d\vec{t}. \label{Adetcue}
\end{split}
\end{equation}
Inserting these results into formula \eqref{akevern} gives
\begin{equation}
\begin{split}
&\mathbb{E}\left(\det(T-x)^{k}\det(T^{\dagger}-y)^{k}\right) \\
&= \frac{\prod_{j=0}^{k-1}h_{j+N}}{\prod_{j=0}^{k-1}(j!)^{2}}x^{k(N+k)}\frac{1}{(\pi\Gamma(\kappa))^{k}}\det\bigg\{\frac{\partial^{i+j-2}A^{(\kappa)}_{N+k}(x,y)}{\partial x^{i-1}\partial y^{j-1}}\bigg\}_{i,j=1}^{k}\\
&= \frac{1}{C^{(\mathrm{JUE}_{\kappa+N,0})}_{k}}\int_{[0,1]^{k}}\prod_{j=1}^{k}t_{j}^{\kappa}(1+(xy-1)t_{j})^{N}\,\Delta^{2}(\vec{t})\,d\vec{t}. \label{lastlinetcue}
\end{split}
\end{equation}
To obtain the last line above we inserted \eqref{Adetcue} and identified the pre-factors as
\begin{equation}
\frac{\prod_{j=0}^{k-1}h_{j+N}}{\prod_{j=0}^{k-1}(j!)^{2}}\frac{(-1)^{\frac{k(k-1)}{2}}}{(\pi\Gamma(\kappa))^{k}}\,e_{N,k,\kappa} = \frac{1}{C^{(\mathrm{JUE}_{\kappa+N,0})}_{k}}.
\end{equation}
Initially we assumed $x>0$ and $y>0$, but since \eqref{lastlinetcue} is an identity between polynomials in $x$ and $y$, it holds on the whole complex plane. This proves Proposition \ref{prop:JUEduality}. Now Theorem \ref{th:tcuefiniten} follows by setting $x = \overline{y} = z$ and the change of variables $t_{j} \to t_{j}/(1-|z|^{2})$ (assuming $|z|<1$) which expresses the right-hand side of \eqref{jueduality} as the distribution function of the largest eigenvalue in the Jacobi Unitary Ensemble. The explicit form of the constant in \eqref{r2kzero} follows by setting $x=y=0$ in \eqref{jueduality}. For the boundary limit, note that if $|z|=1-\frac{u}{N}$ and $\kappa = M-N$ is fixed, we have the limit
\begin{equation}
\begin{split}
&\lim_{N \to \infty}\int_{[0,1]^{k}}\prod_{j=1}^{k}t_{j}^{\kappa}\left(1+\left(-\frac{2u}{N}+\frac{u^{2}}{N^{2}}\right)t_{j}\right)^{N}\,\Delta^{2}(\vec{t})\,d\vec{t}\\
&=\int_{[0,1]^{k}}\prod_{j=1}^{k}t_{j}^{\kappa}e^{-2ut_{j}}\,\Delta^{2}(\vec{t})\,d\vec{t}\\
&=\frac{G(2+k)G(\kappa+k+1)}{G(\kappa+1)}(2u)^{-k^{2}-k\kappa}F_{k+\kappa,k}(2u)
\end{split}
\end{equation}
where $F_{k+\kappa,k}(x)$ is the distribution function of the largest eigenvalue in the Laguerre Unitary Ensemble with parameter $\kappa$. To complete the proof of Theorem \ref{th:tcueonecharge} it is enough to note that the constant in Proposition \ref{prop:JUEduality} has the asymptotics
\begin{equation}
C^{(\mathrm{JUE}_{\kappa+N,0})}_{k} = G(1+k)G(2+k)N^{-k^{2}}\left(1+o(1)\right), \qquad N \to \infty.
\end{equation}
\end{proof}
\begin{remark}
In the limiting case $\kappa \to 0$ and $x=y=1$, the right-hand side of \eqref{jueduality} becomes completely explicit and re-derives the exact formula for moments of the characteristic polynomial of a Haar distributed unitary matrix due to Keating and Snaith \cite{KS00}.
\end{remark} 

\subsection{Planar orthogonal polynomials and reduction to integrals over $U(N)$}
We now consider the quantity \eqref{tcuemoments} with $k=\frac{\gamma}{2}$ and $\gamma$ a general real parameter up to the integrability constraint $\gamma>-2$. As with the Ginibre case, this average over non-Hermitian matrices can be reduced to an integral over the unitary group.
\begin{theorem}
\label{prop:tPVI}
Let $\gamma > -2$ and let $d\mu(U)$ denote the normalized Haar measure on the group of $N \times N$ unitary matrices. Then we have the identity
\begin{equation}
\begin{split}
R_{\gamma}(z) &= R_{\gamma}(0)\int_{U(N)}\det(U)^{-\frac{\gamma}{4}}|\det(I+U)|^{\frac{\gamma}{2}}\det(I+|z|^{2}U)^{\kappa+\frac{\gamma}{2}}\,d\mu(U)\\
&= \frac{R_{\gamma}(0)}{N!(2\pi)^{N}}\int_{[-\pi,\pi]^{N}}\prod_{j=1}^{N}e^{-\frac{i\gamma \theta_j}{4}}|1+e^{i\theta_j}|^{\frac{\gamma}{2}}(1+|z|^{2}e^{i\theta_j})^{\kappa+\frac{\gamma}{2}}|\Delta(e^{i\theta})|^{2}\,d\vec{\theta} \label{tcuetocue}
\end{split}
\end{equation}
where the constant is
\begin{equation}
R_{\gamma}(0) = \prod_{j=0}^{N-1}\frac{\Gamma(\frac{\gamma}{2}+j+1)\Gamma(j+\kappa+1)}{\Gamma(j+1)\Gamma(\frac{\gamma}{2}+j+\kappa+1)}. \label{tcueprefactor}
\end{equation}
Furthermore, we have the exact representation in terms of the Jimbo-Miwa-Okamoto $\sigma$-form of Painlev\'e VI:
\begin{equation}
\label{tcuepvirep}
R_{\gamma}(z) = R_{\gamma}(1)\,\mathrm{exp}\left(\int_{|z|^{2}}^{1}\frac{h(t)-e_{2}'t+\frac{1}{2}e_{2}}{t(1-t)}\,dt\right),
\end{equation}
where $h$ satisfies the equation
\begin{equation}
h'(t(1-t)h'')^{2}+(h'(2h-(2t-1)h')+\tilde{b}_{1}\tilde{b}_{2}\tilde{b}_{3}\tilde{b}_{4})^{2} = \prod_{j=1}^{4}(h'+\tilde{b}_{j}^{2}) \label{heqn}
\end{equation}
with parameters
\begin{equation}
\begin{split}
\tilde{b}_{1} &= \frac{\kappa+N}{2}, \qquad \hspace{8pt}\tilde{b}_{2} = \frac{\kappa+\gamma+N}{2},\\
\tilde{b}_{3} &= \frac{-\kappa+N}{2}, \qquad \tilde{b}_{4} = -\frac{N+\gamma+\kappa}{2}. \label{paramsbtilde}
\end{split}
\end{equation}
In \eqref{tcuepvirep}, $e_{2}$ (respectively $e_{2}'$) is the elementary symmetric polynomial of degree $2$ in $(\tilde{b}_{1},\tilde{b}_{2},\tilde{b}_{3},\tilde{b}_{4})$ (respectively in $(\tilde{b}_{1},\tilde{b}_{3},\tilde{b}_{4})$) and the constant $R_{\gamma}(1)$ is explicit, see equation \eqref{morris}.
\end{theorem}

\begin{proof}
By rotational symmetry the function $R_{\gamma}(z)$ only depends on $|z|$, so without loss of generality we assume $z>0$ throughout this proof. We will again start from orthogonal polynomials in the complex plane and apply the Green's theorem trick to reduce the planar orthogonality to contour orthogonality over the boundary of the unit disc. We proceed in a similar vein to the proof of Theorem \ref{th:ginibretocue}. Let $p_{j}(\lambda)$ denote the monic degree $j$ polynomials satisfying the planar orthogonality ($k \leq j$)
\begin{equation}
h_{j}\delta_{j,k} = \int_{\mathbb{D}}p_{j}(\lambda)\overline{\lambda}^{k}|\lambda-z|^{\gamma}(1-|\lambda|^{2})^{\kappa-1}\,d^{2}\lambda. \label{tcueorthog}
\end{equation}
The goal is to obtain an analogue of the Ginibre identity \eqref{ginibrecontour}. First, using the orthogonality to replace $\overline{\lambda}^{k}$ with $(\overline{\lambda}-z)^{k}$, the orthogonality \eqref{tcueorthog} is easily written in the differential form
\begin{equation}
h_{j}\delta_{k,j} = \int_{\mathbb{D}}p_{j}(\lambda)(\lambda-z)^{\frac{\gamma}{2}}\frac{\partial}{\partial \overline{\lambda}}h(\lambda,\overline{\lambda},z)\,d^{2}\lambda, \label{diff-form}
\end{equation}
where
\begin{equation}
\begin{split}
h(\lambda,\overline{\lambda},z) &:= \int_{z}^{\overline{\lambda}}(s-z)^{\frac{\gamma}{2}+k}(1-\lambda s)^{\kappa-1}\,ds\\
&= (\overline{\lambda}-z)^{\frac{\gamma}{2}+k+1}(1-\lambda z)^{\kappa-1}\int_{0}^{1}s^{\frac{\gamma}{2}+k}(1-s)^{\kappa-1}\,ds.
\end{split}
\end{equation}
In the following we denote the above integral (a particular case of the Euler beta integral) by
\begin{equation}
c_{\gamma,k,\kappa} :=\int_{0}^{1}s^{\frac{\gamma}{2}+k}(1-s)^{\kappa-1}\,ds = \frac{\Gamma(\frac{\gamma}{2}+k+1)\Gamma(\kappa)}{\Gamma(\frac{\gamma}{2}+k+\kappa+1)}.
\end{equation}
Now we apply Green's theorem to \eqref{diff-form} resulting in the boundary integral
\begin{equation}
h_{j}\delta_{k,j} = \frac{1}{2i}\oint_{S^{1}}p_{j}(\lambda)(\lambda-z)^{\frac{\gamma}{2}}h(\lambda,\overline{\lambda},z)\,d\lambda. \label{tcues1integral}
\end{equation}
To simplify this orthogonality further, we recall that we can form linear combinations on both sides of \eqref{tcues1integral} without altering the orthogonality. Indeed, defining
\begin{equation}
c_{p,k}(z) := \binom{p}{k}\,z^{p-k}\frac{c_{\gamma,p,\kappa}}{c_{\gamma,k,\kappa}}
\end{equation}
we multiply both sides of \eqref{tcues1integral} by $c_{p,k}(z)$ and sum from $k=0$ to $k=p$ for some $p \leq j$. On the right-hand side this leads to the sum
\begin{equation}
\sum_{k=0}^{p}c_{p,k}(z)(\overline{\lambda}-z)^{k}c_{\gamma,k,\kappa} = \overline{\lambda}^{p}c_{\gamma,p,\kappa}
\end{equation}
and the analogue of identity \eqref{ginibrecontour} follows,
\begin{equation}
\begin{split}
h_{j}\delta_{k,j} &:= \int_{\mathbb{D}}p_{j}(\lambda)\overline{\lambda}^{k}|\lambda-z|^{\gamma}(1-|\lambda|^{2})^{\kappa-1}\,d^{2}\lambda\\
&=\pi \frac{\Gamma(\frac{\gamma}{2}+j+1)\Gamma(\kappa)}{\Gamma(\frac{\gamma}{2}+j+\kappa+1)}\,\oint_{S^{1}}p_{j}(\lambda)\lambda^{-k}|\lambda-z|^{\gamma}(1-z\lambda)^{\kappa}\frac{d\lambda}{2\pi i \lambda}. \label{tcuecontour}
\end{split}
\end{equation}

The same reasoning that led to \eqref{tdetgin} gives, in this context, the Toeplitz determinant expression
\begin{equation}
\begin{split}
R_{\gamma}(z) =& \frac{N!}{\mathcal{Z}^{(\mathrm{tCUE})}_{N}}\prod_{j=0}^{N-1}h_{j}=\frac{N!}{\mathcal{Z}^{(\mathrm{tCUE})}_{N}}\prod_{j=0}^{N-1}\left(\pi \frac{\Gamma(\frac{\gamma}{2}+j+1)\Gamma(\kappa)}{\Gamma(\frac{\gamma}{2}+j+\kappa+1)}\right)\\
&\times\det\bigg\{\oint_{S^{1}}\lambda^{j-i}|\lambda-z|^{\gamma}(1-z\lambda)^{\kappa}\,\frac{d\lambda}{2\pi i \lambda}\bigg\}_{i,j=0}^{N-1},
\end{split}
\end{equation}
where $\mathcal{Z}_{N}^{(\mathrm{tCUE})}$ is the normalization constant for \eqref{zspdf} which can be calculated explicitly (see \textit{e.g.} \cite{ZS00}). This yields the form of the constant pre-factor
\begin{equation}
\frac{N!}{\mathcal{Z}^{(\mathrm{tCUE})}_{N}}\,\prod_{j=0}^{N-1}\left(\pi\frac{\Gamma(\frac{\gamma}{2}+j+1)\Gamma(\kappa)}{\Gamma(\frac{\gamma}{2}+j+\kappa+1)}\right) = R_{\gamma}(0), \label{rgamma0ident}
\end{equation}
an identity which follows from computing $R_{\gamma}(0)$ via \eqref{norms} and formula \eqref{hjtcue} with $j \to j+\frac{\gamma}{2}$.

Making appropriate changes of variables and deformations of the contour, the integrals inside the determinant can be written in the equivalent form
\begin{equation}
\oint_{S^{1}}\lambda^{j-i}|\lambda-z|^{\gamma}(1-z\lambda)^{\kappa}\,\frac{d\lambda}{2\pi i \lambda} = (-z)^{j-i}\oint_{S^{1}}\lambda^{j-i}\,\lambda^{-\frac{\gamma}{4}}|1+\lambda|^{\frac{\gamma}{2}}(1+z^{2}\lambda)^{\kappa+\frac{\gamma}{2}}\,\frac{d\lambda}{2\pi i \lambda},
\end{equation}
where we used identity \eqref{idabs}. Now \eqref{tcuetocue} follows from the same considerations given in the proof of Theorem \ref{th:ginibretocue}. The association of the group integrals in \eqref{tcuetocue} with Painlev\'e VI has been studied by Forrester and Witte, for example see Equations (1.6), (1.21), (1.34) and Proposition 13 in \cite{FWPVI}, from which identity \eqref{tcuepvirep} follows. 
\end{proof}

\begin{corollary}
\label{cor:consist}
The identity \eqref{JUEprobab} can be extended to any real $k = \frac{\gamma}{2} > -1$ by continuing the integer valued parameters in the equation \eqref{JUEprobab}. Equivalently, the Painlev\'e characterization \eqref{tcuepvirep} reduces to the one in \eqref{JUEprobab} when $\gamma=2k$.
\end{corollary}

\begin{proof}
First notice that identity \eqref{JUEprobab} can easily be written
\begin{equation}
R_{2k}(z) = R_{2k}(0)\,\mathrm{exp}\left(-\int_{1-|z|^{2}}^{1}\frac{\sigma^{(\mathrm{VI})}_{\kappa,N}(t)-(b_1 b_2-k\kappa-k^{2})t+\frac{b_{1}b_{2}+b_{3}b_{4}}{2}-k\kappa-k^{2}}{t(1-t)}\,dt\right).
\end{equation}
Regarding identity \eqref{tcuepvirep}, some simple changes of variables casts it in the form
\begin{equation}
R_{\gamma}(z) = R_{\gamma}(0)\,\mathrm{exp}\left(-\int_{1-|z|^{2}}^{1}\frac{h(1-t)+e_{2}'t-e_{2}'+\frac{1}{2}e_{2}}{t(1-t)}\,dt\right).
\end{equation}
Now it is straightforward to show that $h(t)$ solves equation \eqref{heqn} if and only if $h(1-t)$ solves \eqref{pvieqn} replacing everywhere $k$ with $\frac{\gamma}{2}$ and identifying $b_{1},\ldots,b_{4}$ in terms of $\alpha = \kappa$ and $\beta=N$, see \eqref{paramsb} and \eqref{paramsbtilde}. It remains to observe the simple identities
\begin{equation}
\begin{split}
\left(b_{1}b_{2}-k\kappa-k^{2}\right)|_{k=\frac{\gamma}{2}} &= -e_{2}'\\
\left(\frac{b_{1}b_{2}+b_{3}b_{4}}{2}-k\kappa-k^{2}\right)\bigg{|}_{k=\frac{\gamma}{2}}& = \frac{1}{2}e_{2}-e_{2}'
\end{split}
\end{equation}
and so the characterization \eqref{JUEprobab} remains true for any real values $k=\frac{\gamma}{2}$ with $k > -1$.
\end{proof}

\begin{remark}
As with the Ginibre case, the identity \eqref{tcuecontour} implies that the planar orthogonal polynomials are expressible in terms of certain polynomials orthogonal on the unit circle. These polynomials have been characterized in terms of fundamental objects associated with the Painlev\'e VI system, see \cite{FW06}. Discrete $N$-recurrences for the group integrals in \eqref{tcuepvirep} were derived in \cite{FW04}. As before, identities \eqref{tcuetocue} and \eqref{jueduality} (with $x=y=|z|$) imply a duality between averages over the CUE and the JUE, which has been discussed in the context of generalized hypergeometric functions in Section 3.3 of \cite{FWPVI}. Its appearance here is a new manifestation of the duality arising from computing an average of truncated unitary matrices in two different ways. Interestingly, averages of this type have been shown to arise in last passage percolation and the Ising model, see Section 5.3 of \cite{FWPVI}.
\end{remark}
\begin{remark}
The constant $R_{\gamma}(1)$ in \eqref{tcuepvirep} can be computed explicitly from \eqref{tcuetocue} by setting $z=1$ and reinstating the absolute value (see \textit{e.g.} identity \eqref{idabs}). This leads to the Morris integral evaluation,
\begin{equation}
\begin{split}
R_{\gamma}(1) &= \frac{R_{\gamma}(0)}{N!(2\pi)^{N}}\int_{[-\pi,\pi]^{N}}\prod_{j=1}^{N}e^{\frac{i\kappa \theta_j}{2}}|1+e^{i\theta}|^{\gamma+\kappa}|\Delta(e^{i\theta})|^{2}\,d\vec{\theta}\\
&= R_{\gamma}(0)\prod_{j=0}^{N-1}\frac{\Gamma(1+\kappa+\gamma+j)\Gamma(1+j)}{\Gamma(1+\kappa+\frac{\gamma}{2}+j)\Gamma(1+\frac{\gamma}{2}+j)}\\
&= \prod_{j=0}^{N-1}\frac{\Gamma(\kappa+j+1)\Gamma(\kappa+\gamma+j+1)}{\Gamma(\kappa+\frac{\gamma}{2}+j+1)^{2}}, \label{morris}
\end{split}
\end{equation}
where we used \eqref{tcueprefactor}. Writing this in terms of Barnes-G functions and inserting their asymptotic expansion (see \cite[Eq. 5.17.5]{NIST:DLMF}) gives
\begin{equation}
R_{\gamma}(1) = N^{\frac{\gamma^{2}}{4}}\,\frac{G(\kappa+\frac{\gamma}{2}+1)^{2}}{G(\kappa+1)G(\kappa+\gamma+1)}\,(1+o(1)), \qquad N \to \infty,
\end{equation}
which is consistent with the limit $u \to 0$ of \eqref{tcueresult} with $k=\frac{\gamma}{2}$.
\end{remark}

\subsection{Painlev\'e VI to Painlev\'e V scaling heuristics}
Starting from the identity \eqref{JUEprobab}, we consider the boundary scaling $|z|^{2} = 1-\frac{2u}{N}$ in the weak non-unitarity limit, so that $\kappa := M-N$ is fixed. A change of variable $t \to t/N$ in \eqref{pvirep} gives the characterization
\begin{equation}
\mathbb{P}\left(\lambda^{(\mathrm{JUE}_{k,\kappa,N})}_{\mathrm{max}} < \frac{2u}{N}\right) =\mathrm{exp}\left(-\int_{2u}^{N}\frac{\sigma^{(\mathrm{VI})}_{\kappa,N}(\frac{t}{N})-b_{1}b_{2}\frac{t}{N}+\frac{b_{1}b_{2}+b_{3}b_{4}}{2}}{t(1-\frac{t}{N})}\,dt\right)
\end{equation}
Defining the centered and scaled function
\begin{equation}
v(t) = \sigma^{(\mathrm{VI})}_{\kappa,N}\left(\frac{t}{N}\right)-b_{1}b_{2}\frac{t}{N}+\frac{b_{1}b_{2}+b_{3}b_{4}}{2}
\end{equation}
it is easily shown that $v$ satisfies the equation
\begin{equation}
\begin{split}
&N^{3}\left(v'+\frac{b_{1}b_{2}}{N}\right)t^{2}\left(1-\frac{t}{N}\right)^{2}(v'')^{2}\\
&-\left(N(N-2t)(v')^{2}+(2vN+(b_{1}b_{2}-b_{3}b_{4})N-2b_{1}b_{2}t)v'+2vb_{1}b_{2}\right)^{2}\\
&+N^{4}\prod_{i=1}^{4}\left(v'+\frac{b_{1}b_{2}-b_{i}^{2}}{N}\right)=0. \label{rescaledPVI}
\end{split}
\end{equation}
Now dividing both sides of \eqref{rescaledPVI} by $N^{4}$ and, using the explicit form of the $N$-dependent parameters $b_{1},\ldots,b_{4}$ in \eqref{paramsb} (recall that $\alpha=\kappa$ and $\beta=N$), the limit $N \to \infty$ gives the equation
\begin{equation}
(tv'')^{2}-\left(2(v')^{2}-tv'+(2k+\kappa)v'+v\right)^{2}+4(v')^{2}(v'+k)(v'+\kappa+k) = 0,
\end{equation}
which is precisely the $\sigma$-form of Painlev\'e V \eqref{pvintro} with parameter $\alpha=\kappa$. Thus
\begin{equation}
\lim_{N \to \infty}\mathbb{P}\left(\lambda^{(\mathrm{JUE}_{k,\kappa,N})}_{\mathrm{max}} < \frac{2u}{N}\right) = \mathrm{exp}\left(-\int_{2u}^{\infty}\frac{v(t)}{t}\,dt\right).
\end{equation}
Clearly these scaling heuristics are consistent with Theorem \ref{th:tcueonecharge}.

\appendix
\section{Relation to central limit theorems and the Gaussian free field}
\label{app:gff}
In this appendix we will discuss the asymptotic expansion in Theorem \ref{th:onecharge} when $z$ is in the exterior region $|z|>1+\delta$ where $\delta>0$ is fixed. As with the proof of Theorem \ref{th:onecharge}, we take the duality \ref{prop:lue} as our starting point. In particular, rescaling equation \eqref{LUEduality} via $t_{j} \to t_{j}\sqrt{N}$ for each $j=1,\ldots,k$ gives the identity
\begin{equation}
\mathbb{E}\left(|\det(G_{N}-z)|^{2k}\right) = \frac{N^{k^{2}}}{\prod_{j=0}^{k-1}(j!(j+1)!)}I^{(2k)}_{N}(z) \label{ink}
\end{equation}
where
\begin{equation}
I^{(2k)}_{N}(z) := \int_{[0,\infty)^{k}}\prod_{j=1}^{k}e^{-N\phi(t_j)}\Delta^{2}(\vec{t})\,d\vec{t}
\end{equation}
and
\begin{equation}
\phi(t) = t-\log(|z|^{2}+t).
\end{equation}
For $|z|>1$ fixed, the function $\phi(t)$ achieves its minimum at $t=0$, the endpoint of integration. We have,
\begin{equation}
\phi(t) = -2\log|z|+t(1-|z|^{-2})+O(t^{2}), \qquad t \to 0.
\end{equation}
Therefore by classical saddle point arguments we get the approximation
\begin{equation}
\begin{split}
&I^{(2k)}_{N}(z)\sim |z|^{2Nk}\int_{[0,\infty)^{k}}\prod_{j=1}^{k}e^{-Nt_{j}(1-|z|^{-2})}\,\Delta^{2}(\vec{t})\,d\vec{t}\\
&=|z|^{2Nk}N^{-k^{2}}(1-|z|^{-2})^{-k^{2}}\prod_{j=0}^{k-1}j!(j+1)! \label{outsideapprox}
\end{split}
\end{equation}
where we used the Selberg integral type formula (see \textit{e.g.} \cite[Chapter 4]{Forrester_loggas})
\begin{equation}
\int_{[0,\infty)^{k}}\prod_{j=1}^{k}e^{-t_{j}}\,\Delta^{2}(\vec{t})\,d\vec{t} = \prod_{j=0}^{k-1}j!(j+1)!.
\end{equation}
Now inserting \eqref{outsideapprox} into \eqref{ink} we obtain the asymptotic expansion
\begin{equation}
\mathbb{E}\left(|\det(G_{N}-z)|^{2k}\right) = e^{2Nk\log|z|-k^{2}\log(1-|z|^{-2})}\,\left(1+o(1)\right), \qquad N \to \infty. \label{zoutside}
\end{equation}

The expansion \eqref{zoutside} could also be arrived at from the global fluctuation theory of smooth linear statistics studied in \cite{RV07,AHM11}. In this sense, we have the interpretation
\begin{equation}
\mathbb{E}\left(|\det(G_{N}-z)|^{\gamma}\right) = \mathbb{E}(e^{\gamma L_{N}[f_{z}]})
\end{equation}
where
\begin{equation}
L_{N}[f_z] := \log|\det(G_{N}-z)| = \sum_{j=1}^{N}f_{z}(\lambda_j), \qquad f_{z}(\lambda) := \log|\lambda-z|, \label{logrv}
\end{equation}
and if $|z|>1$, $f_{z}(\lambda)$ is a smooth function of $\lambda$ on the eigenvalue support. In fact consider a general \textit{linear statistic}
\begin{equation}
L_{N}[f] := \sum_{j=1}^{N}f(\lambda_j)
\end{equation}
where $\lambda_j$ are the $N$ eigenvalues of the complex Ginibre random matrix $G_{N}$. Then for sufficiently smooth $f$, it holds that
\begin{equation}
\mathbb{E}(e^{\gamma L_{N}[f]}) = e^{N\gamma m_{1}(f)+\frac{\gamma^{2}}{2}m_{2}(f)}\left(1+o(1)\right), \qquad N \to \infty, \label{CLTgin}
\end{equation}
where the asymptotic mean is
\begin{equation}
m_{1}(f) = \frac{1}{\pi}\int_{\mathbb{D}}\,f(\lambda)\,d^{2}\lambda
\end{equation}
and the asymptotic variance is given by
\begin{equation}
m_{2}(f) = \frac{1}{4\pi}\int_{\mathbb{D}}|\nabla f(\lambda)|^{2}\,d^{2}\lambda + \frac{1}{2}\sum_{k=-\infty}^{\infty}|k||\hat{f}(k)|^{2} \label{gradsquared}
\end{equation}
where $\hat{f}(k) = \frac{1}{2\pi}\int_{0}^{2\pi}f(e^{i\theta})e^{-ik\theta}\,d\theta$. The gradient squared term in \eqref{gradsquared} is associated with Gaussian free field fluctuations, again see \cite{RV07} for discussion about this. In the particular case of $f_{z}(\lambda) = \log|\lambda-z|$, $|z|>1$, the mean and variance can be computed explicitly from these formulae as $m_{1} = \log|z|$ and $m_{2} = -\frac{1}{2}\log(1-|z|^{-2})$. Setting $k=\frac{\gamma}{2}$ then shows the agreement between \eqref{zoutside} and \eqref{CLTgin}, leading to expansion \eqref{supercrit}. On the other hand when $|z| \leq 1$ or $z$ is too close to the boundary, $f_{z}(\lambda)$ lacks smoothness on the support of the eigenvalues and the asymptotic expansion \eqref{CLTgin} takes a more complicated form, as Theorems \ref{th:ww} and \ref{th:onecharge} demonstrate.

\section*{Acknowledgements}
The authors would like to thank Gernot Akemann, Yan V. Fyodorov, Arno B. J. Kuijlaars and Seung-Yeop Lee for helpful discussions on the results of this work. We thank Peter J. Forrester for comments and for pointing out reference \cite{LJP99} to us. We are very grateful to the anonymous referees for their careful reading and comments on the paper. A. D. gratefully acknowledges financial support from EPSRC, First Grant project  ``Painlev\'e equations: analytical properties and numerical computation", reference EP/P026532/1. N. S. gratefully acknowledges support of the Royal Society University Research Fellowship “Random matrix theory and log-correlated Gaussian fields”, reference URF\textbackslash R1\textbackslash 180707.

\bibliographystyle{plain}
\bibliography{non-herm}

\end{document}